\documentclass[journal]{IEEEtran}

\usepackage{cite}
\usepackage{amsmath,amssymb,amsfonts}
\usepackage{marvosym}
\usepackage{algorithmic}
\usepackage{algorithm}
\usepackage{graphicx}
\usepackage{xcolor}
\usepackage{array}
\usepackage[caption=false,font=footnotesize,labelfont=rm,textfont=rm]{subfig}
\usepackage{textcomp}
\usepackage{stfloats}
\usepackage{url}
\usepackage{multirow}
\usepackage{verbatim}
\usepackage{float}  
\usepackage{overpic}
\usepackage{amsthm}
\usepackage{tabularx,booktabs,makecell,adjustbox}
\renewcommand\arraystretch{1.05}

\theoremstyle{plain} 
\newtheorem{theorem}{Theorem}[section]
\newtheorem{lemma}[theorem]{Lemma}

\theoremstyle{definition}

\theoremstyle{remark}

\begin{document}

\title{Plasticity-Enhanced Multi-Agent Mixture of Experts for Dynamic Objective Adaptation in UAVs-Assisted Emergency Communication Networks}

\author{Wen Qiu, \IEEEmembership{Member,~IEEE}, Zhiqiang He, \IEEEmembership{Member,~IEEE}, Wei Zhao, \IEEEmembership{Member,~IEEE}, Hiroshi Masui

\thanks{Wen Qiu is with the Department of Information and Communication Engineering, Kitami Institute of Technology, Japan. Email: w-qiu@ieee.org.}

\thanks{Zhiqiang He is with the Graduate School of Informatics and Engineering, the University of Electro-Communications, Japan. Email: hezhiqiang@ieee.org.}

\thanks{Wei Zhao is with the School of Computer Science and Technology, Anhui University of Technology, China. Emails: zhaowei@ahut.edu.cn.}
    
\thanks{Hiroshi Masui is with the Department of Information and Communication Engineering
and also Information Procession Center, Kitami Institute of Technology, Japan. Email: hgmasui@mail.kitami-it.ac.jp.}

\thanks{The corresponding author are Wei Zhao and Hiroshi Masui.}
}

\maketitle

\begin{abstract}
Unmanned aerial vehicles serving as aerial base stations can rapidly restore connectivity after disasters, yet abrupt changes in user mobility and traffic demands shift the quality of service trade-offs and induce strong non-stationarity. Deep reinforcement learning policies suffer from plasticity loss under such shifts, as representation collapse and neuron dormancy impair adaptation. We propose plasticity enhanced multi-agent mixture of experts (PE-MAMoE), a centralized training with decentralized execution framework built on multi-agent proximal policy optimization. PE-MAMoE equips each UAV with a sparsely gated mixture of experts actor whose router selects a single specialist per step. A non-parametric Phase Controller injects brief, expert-only stochastic perturbations after phase switches, resets the action log-standard-deviation, anneals entropy and learning rate, and schedules the router temperature, all to re-plasticize the policy without destabilizing safe behaviors. We derive a dynamic regret bound showing the tracking error scales with both environment variation and cumulative noise energy. In a phase-driven simulator with mobile users and 3GPP-style channels, PE-MAMoE improves normalized interquartile mean return by 26.3\% over the best baseline, increases served-user capacity by 12.8\%, and reduces collisions by approximately 75\%. Diagnostics confirm persistently higher expert feature rank and periodic dormant-neuron recovery at regime switches.
\end{abstract}

\begin{IEEEkeywords}
unmanned aerial vehicles, emergency communications, multiagent reinforcement learning, mixture of experts, continuous learning, non-stationary environments, resource allocation.
\end{IEEEkeywords}

\section{Introduction}
\IEEEPARstart{I}{n} the wake of natural disasters such as earthquakes, floods, or wildfires, terrestrial communication infrastructures are often knocked out. This collapse of connectivity not only delays the delivery of critical information but also hampers coordinated emergency response efforts. Recent advances in unmanned aerial vehicle (UAV)-assisted communication have introduced aerial base stations as a promising means of rapidly restoring coverage in such disrupted environments~\cite{11047530, SUN2023102200}. Enabled by their flexibility and line of sight propagation advantages, fleets of UAVs can dynamically reposition to serve mobile users across challenging terrains~\cite{wu2018joint}. However, achieving robust and efficient UAV coordination remains nontrivial in the presence of user mobility and dynamic demand shifts.

Multi-agent reinforcement learning (MARL) has emerged as a powerful tool for tackling complex coordination problems in UAV swarms, such as trajectory planning~\cite{foerster2018counterfactual}, energy aware task allocation~\cite{zhang2024distributed}, and real time spectrum sharing~\cite{badnava2021spectrum, yin2025deepthinkvla}. These methods typically rely on predefined reward structures to encode the optimization objectives, for instance maximizing user coverage or minimizing transmission latency. Yet in real world emergencies, users move and their communication demands evolve over time, which shifts the objective weights across quality of service (QoS) terms such as coverage, energy, and collision avoidance, thereby making the learning problem inherently non-stationary. Evidence from UAVs-assisted networking shows that time varying, spatially non-uniform traffic and user mobility are central drivers of these reweightings and must be explicitly modeled in policy design~\cite{9796953,lee2025slowsteadywinsrace}. These evolving priorities render the learning environment inherently non-stationary.

Although mixture of experts (MoE) architectures~\cite{shazeer2017outrageously,willi2024mixture} and gradient noise injection techniques~\cite{nikishin2022primacy,NEURIPS2024_978cc34c} have shown promise in preserving plasticity in streaming or continual learning contexts, their potential remains largely underexplored in UAVs-based multi-agent systems. In MARL, prevailing approaches to non-stationarity fall into three broad categories. (1) Replay and stabilization methods attach extra context to samples or down weight stale data, for example experience-replay “fingerprinting” and multi-agent importance sampling to make off policy updates usable~\cite{Foerster2017StabilisingER}. (2) Methods that model or shape opponents and teammates explicitly anticipate other learners’ updates; notable examples include learning with opponent-learning awareness (LOLA) and proximal LOLA variants, which improve coordination but presume access to or accurate prediction of others’ learning dynamics~\cite{Foerster2018LOLA, Zhao2022POLA}. (3) Meta RL learns priors or update rules for quick task switching, which requires training over representative task distributions and can be brittle when mission phases deviate from the meta train manifold ~\cite{finn2017maml}. UAV-specific work typically tackles dynamics via receding horizon, heuristic replanning, or standard deep reinforcement learning (DRL) controllers tuned per scenario, with limited treatment of abrupt objective re-weighting and representation degradation during long runs. Recent surveys highlight that while DRL and MARL are increasingly used for trajectory and resource control, robust handling of phase level nonstationarity remains a gap~\cite{Li2024MARLSurvey, 11047530}. 

In this paper, we introduce the plasticity enhanced multi-agent mixture of experts (PE-MAMoE), a MARL framework designed for rapidly evolving UAV-assisted communication tasks. Our approach diverges from the above in two ways. First, instead of opponent aware updates or off policy replay fixes, we target the representation failure mode directly. PE-MAMoE preserves expressivity in a sparsely gated MoE actor via expert only stochastic perturbations injected briefly after phase switches, combined with lightweight router temperature and entropy scheduling to prevent expert under utilization. The entire mechanism is fully compatible with on-policy multi-agent PPO (MAPPO). Second, rather than meta training on many tasks, we operate within a single, phase driven ECN where weights over QoS terms change online. Our method maintains plasticity so the same policy can track shifting optima under user mobility and time varying demand, filling the UAV MARL gap where prior methods either assume stationary objectives or lack mechanisms to prevent rank collapse and neuron dormancy during long term operation~\cite{PrudentQLearning2022}. 

We further provide a theoretical analysis of PE\text{-}MAMoE, deriving a dynamic regret bound that characterizes the trade-off between active forgetting and knowledge retention in non-stationary environments. The bound shows that the tracking error scales with both the path length of the evolving optimal objective sequence and the cumulative noise energy from stochasticity injection, offering principled guidance for tuning injection schedules. We validate our framework in a phase-driven simulator with mobile users, fluctuating interference, and 3GPP-style channel modeling. Compared with strong baselines, multilayer perceptron (MLP), MoE, and Sparse MoE, PE\text{-}MAMoE improves normalized interquartile mean (IQM) return by \textbf{+26.3\%} over the best baseline, increases served-user capacity by \textbf{+12.8\%}, and reduces collisions by \textbf{$\approx$75\%}. It also maintains the highest expert feature rank with periodic recovery and exhibits phase-synchronous troughs in dormant neuron fraction, indicating effective re-plasticization under sparse MoE routing.

The main contributions of this paper are as follows.
\begin{itemize}
\item \textbf{Plasticity enhanced MARL for UAV–ECNs.} We introduce PE\mbox{-}MAMoE, a MAPPO-based controller that marries sparsely gated mixture of experts with controlled stochasticity to maintain plasticity and avoid expert collapse under phase driven non-stationarity. This leverages conditional computation for scalable capacity while remaining stable in cooperative MARL. 
\item \textbf{Realistic, phase driven simulator.} We build a high fidelity UAV–ECN environment with user mobility, demand shifts, clustered reuse and adjacent channel leakage, enabling reproducible evaluation under abrupt objective switches.  
\item \textbf{Plasticity diagnostics and robust evaluation.} We use effective rank and dormant neuron fraction to quantify representation collapse and plasticity across switches, linking empirical behavior to primacy bias remedies via reset and perturb strategies, and using IQM with thorough ablations, show that PE-MAMoE achieves higher return, capacity, and stability than MLP, MoE, and Sparse MoE in fast switching regimes.
\end{itemize}

\section{Related Work}

We structure this section into three subsections to better clarify problem setting, failure mechanisms, and architectural strategy, and to make our contributions explicit by contrast. Multi-UAV coordination for emergency communications describes the application domain and operational objectives that define the evaluation criteria for our setting. Plasticity loss in neural decision systems integrates evidence on representation collapse and dormant units, and explains why extending plasticity preservation from single agent to multi agents learning introduces additional failure modes, non-stationarity from concurrently adapting teammates, coordination conventions that harden too early, and interference under parameter sharing. MoE and Modular Architectures outline modular designs that can mitigate such interference by decoupling capacities and localizing updates. This organization clarifies our contribution: within a highly dynamic UAV–ECN, we introduce PE-MAMoE, a sparsely gated actor with expert local re-plasticization and temperature scheduled routing, integrated with MAPPO, to sustain team level adaptability under abrupt objective and mobility mechanism shifts while avoiding rank collapse and neuron dormancy.

\subsection{Multi-UAV Coordination for Emergency Communications}
UAVs have emerged as a rapid deployment solution for maintaining connectivity when ground infrastructure fails, such as after natural disasters. Acting as aerial base stations, UAVs can swiftly restore wireless coverage in affected areas~\cite{sun2023uav}. Prior work in this domain has tackled core optimization problems including three dimensional (3D) UAV placement and trajectory design, transmit power control, and dynamic user association. For instance, researchers have studied how to jointly optimize UAV flight paths, user-UAV associations, and resource allocation to maximize coverage or network throughput under fairness and energy constraints~\cite{wu2018joint}. These problems are intrinsically complex, both non-convex and combinatorial, and they often demand real time solutions owing to a highly dynamic environment characterized by user mobility and shifting network demands~\cite{10758641, foerster2018counterfactual}. Traditional optimization techniques struggle in this setting, especially without a central controller, motivating the use of RL and multi-agent coordination. Indeed, distributed multi-agent deep reinforcement learning (MADRL) approaches have been proposed to allow UAVs to learn cooperative policies for positioning and resource allocation on the fly~\cite{badnava2021spectrum}. These approaches enable each UAV to act as an agent, adjusting its 3D position and service strategy based on local observations, and have demonstrated superior performance over static or short sighted baselines in terms of achieved throughput and fairness~\cite{hussain2024computing}.

Several broad surveys provide overviews of UAVs-assisted wireless networks and their challenges~\cite{hussain2024computing,zhang2023survey}. Common themes include ensuring robust coverage, interference management, energy efficiency, and integration with existing cellular infrastructure. Wu \emph{et al.} present tutorial discussions on using UAVs for wireless coverage extension in 5G and beyond, highlighting both opportunities and open problems in UAV communications~\cite{wu2018joint}. To address these challenges, a variety of techniques have been explored: from metaheuristic algorithms for UAV placement to model based control for interference mitigation and energy aware scheduling~\cite{sun2023uav}. More recently, single agent and MARL frameworks have been applied to UAV network control. For example, Q learning agents have been used for adaptive UAV deployment and movement, showing the ability of UAVs to learn 3D placements that maximize users’ quality of experience under mobility~\cite{badnava2021spectrum}. Multi-agent policy gradient methods like counterfactual multi-agent policy gradients have also been employed to coordinate multiple UAVs in coverage tasks~\cite{foerster2018counterfactual}. These learning based approaches outperform static heuristics by enabling UAVs to continuously re-position and re-configure according to live network conditions.

However, most existing UAV communication strategies assume a fixed or slowly varying mission objective, for example, maximize coverage of a given area or serve a set of users with known demand. In MARL, mainstream responses to non-stationarity fall into three families. (i) Replay methods attach additional context to samples or down weight stale data so that off policy updates remain usable, for example, ``fingerprinting" identifiers in experience replay to reduce non-stationarity across agents~\cite{Foerster2018LOLA}. (ii) Opponent aware updates explicitly model others’ learning, for example learning with opponent learning awareness, and proximal or consistent variants, which can improve coordination but presume access to or accurate prediction of others’ updates~\cite{Zhao2022POLA}. (iii) Meta RL for rapid adaptation learns priors rules that speed task switching, yet requires representative task families and may degrade when phases deviate from the meta train manifold~\cite{finn2017maml}. Surveys emphasize that dealing with non-stationarity remains a central open issue in MADRL, despite these advances~\cite{Li2024MARLSurvey, 11047530}. In UAV communications specifically, most works react to dynamics via receding horizon re-planning or task specific DRL tuned per scenario, but provide limited treatment of abrupt objective re-weighting and representation collapse during long deployments, a gap repeatedly noted by domain surveys synthesizing UAV networking under mobility, interference, and energy constraints. 

Our approach diverges from the above lines in two ways. First, rather than fixing off-policy instability or modeling opponents, we target the representation failure mode directly. PE-MAMoE preserves expressivity in a sparsely gated MoE actor via short-window, expert-only stochastic perturbations after phase switches, while keeping the router stable and lightly temperature-scheduled. This design is fully compatible with on-policy MAPPO. Second, instead of meta-training across many tasks, we operate within a single, phase-driven ECN where QoS weights change online. The same policy tracks shifting optima under user mobility and time-varying demand by maintaining plasticity, addressing the UAV-MARL gap where prior methods either assume stationary objectives or lack mechanisms to prevent rank collapse and neuron dormancy.

\subsection{Plasticity Loss in Neural Decision Systems}
Continual learning and long running RL agents often suffer from plasticity loss, a reduction in the ability to acquire new knowledge over time~\cite{dohare2024nature}. As training progresses, neural networks tend to gradually and irreversibly lose their plasticity, meaning the model becomes stable but rigid, struggling to adapt to new tasks or changes~\cite{dohare2024nature}. Recent studies have characterized this phenomenon by observing increasing fractions of dormant neurons, units that effectively stop updating or firing, and a collapse in the network’s feature rank or representational diversity~\cite{lyle2023understanding}. In DRL settings, for example, standard agents have been shown to eventually settle into narrow representations: many neurons become inactive and the learned features span a lower dimensional subspace, hampering the agent’s ability to continue learning. This loss of plasticity is closely tied to the stability–plasticity dilemma and has been identified as a culprit behind both catastrophic forgetting and the stagnation of performance on new tasks~\cite{abbas2023loss}.

Researchers have begun quantifying and addressing plasticity loss in neural decision systems. For instance, Sokar \emph{et al.} highlight a “dormant neuron phenomenon” in DRL, where many units drop out of activity during training, correlating with diminished learning capacity~\cite{sokar2023dormant}. Lyle \emph{et al.} further link plasticity loss to a collapse in feature space dimensionality and show how this correlates with the agent becoming over committed to its early policy patterns~\cite{lyle2023understanding}. To counteract these issues, several mitigation strategies have been proposed. One line of attack is meta learning, which trains models to be inherently more adaptable~\cite{finn2017maml,beattie2022metalearning}. By learning an initialization or update rule that remains plastic, meta learning approaches enable quick learning on new tasks, thereby counteracting the tendency to get stuck on past knowledge. Another set of techniques focuses on injecting variability or perturbations during training to avoid the network settling into a complacent state. Nikishin \emph{et al.} identify a primacy bias in deep RL and propose periodically resetting a portion of the agent’s parameters to combat this bias~\cite{nikishin2022primacy}. Similarly, shrink and perturb methods introduce small random perturbations to weights, preserving diversity and preventing dormant neurons~\cite{dohare2024nature}. Despite these advances, most have been studied in single-agent or supervised contexts. Applications to multi-agent coordination problems remain sparse. 

Extending plasticity preserving ideas from single agent to multi agents settings introduces additional failure modes: (i) non-stationarity from concurrently learning teammates or opponents, which destabilizes value and policy updates; (ii) multi-agent credit assignment and convention formation, where agents may overfit to early coordination patterns; (iii) parameter interference under parameter sharing, which can compress representation rank and reduce plasticity; and (iv) architecture specific dormancy in MARL modules that mirrors single agent dormant units. Recent work explicitly documents dormant neurons in MARL value factorization~\cite{qin2024dormant_marl}; surveys and algorithms targeting non-stationarity and coordination further underline these challenges~\cite{papoudakis2019, christianos2021scaling, wang2020roma}. In parallel, continual multi-agent studies show that teams can forget prior coordination skills across regime shifts, pointing to a need for team level plasticity maintenance~\cite{10562331}. Motivated by these findings, we instantiate a plasticity preserving approach for UAVs based multi-agent coordination, PE-MAMoE, that couples modular experts with stochasticity based plasticity injection to sustain adaptability under abrupt user demand and mobility regime shifts, extending plasticity loss remedies beyond single agent and supervised settings to MARL for ECNs.

\subsection{MoE and Modular Architectures}
MoE architectures have grown in popularity as a way to scale up models while handling task diversity. The MoE concept, originally proposed in the 1990s, involves a number of expert subnetworks and a gating mechanism that routes each input to one or a few experts~\cite{jacobs1991adaptive}. Shazeer \emph{et al.} revitalized this idea in deep learning (DL) by introducing a sparsely gated MoE layer that enabled extremely large neural networks with manageable computation~\cite{shazeer2017outrageously}. The key is that only a small subset of experts is active for any given input, so the model’s capacity grows while inference cost remains relatively constant. This architecture offers a natural way to handle heterogeneous tasks or data modes. Recent studies in deep RL have shown that augmenting agents with MoE layers significantly improves learning capacity and reduces dormant neurons~\cite{he2024moe}. Several advances have been made to address specific challenges in MoE systems. One focus area is routing sparsity and efficiency. Switch Transformers simplified the MoE gating by routing each input to exactly one expert, enabling trillion parameter models with moderate cost~\cite{fedus2022switch}. Another concern is expert utilization and diversity. Without intervention, experts may be underused or converge to similar functions, undermining modularity~\cite{zhang2022vidhoc}. Diversity promoting regularization and load balancing techniques have been proposed to ensure distinct and balanced expert usage. A related line of research addresses controlled forgetting and lifelong adaptation in modular networks. For example, Lee \emph{et al.} propose gradual expert updates to preserve old knowledge while learning new tasks~\cite{lee2025slowsteadywinsrace}. 

Despite these advances, integration of MoE architectures with MARL and wireless networking problems remains largely unexplored. To our knowledge, there are no prior studies that use MoEs to handle non-stationary objectives in a multi-UAV coordination scenario. Our work is among the first to introduce a MoE design into MARL for UAV communications, enabling UAV teams to effectively switch strategies on the fly without incurring plasticity loss.

\section{Problem Statement}

We adopt a typical multi-UAV as aerial base station (BS) model over a square disaster area with slotted time and 3D motion characteristics because it is (i) representative of the dominant formulation in UAV communications, capturing the essential trade offs among coverage, interference, and energy, yet (ii) sufficiently compact to admit clean signal-to-interference-plus-noise ratio (SINR) expressions, propulsion and communication power accounting, and a Markov state action interface suited to MAPPO. Meanwhile, the time varying user set $\mathcal{N}(t)$ and user mobility can simulate realistic load and topology changes. Our goal is not to innovate on geometry itself, but to provide a faithful, extensible backbone on which to evaluate team level plasticity under system state shifts. The novelty lies in how the model is used: we embed phase driven QoS objective reweighting and demand and mobility pattern changes into the same environment, yielding abrupt but structured non-stationarity that stress representation plasticity without introducing with any special modeling techniques. This design allows direct comparison to standard baselines by disabling phases while isolating the effect of the proposed PE-MAMoE mechanisms on adaptation, stability, and sample efficiency.

\subsection{System Model}

Table~\ref{tab:notation} summarizes the key symbols used throughout the paper for easy reference.

\begin{table}[!t]
\centering
\scriptsize
\renewcommand{\arraystretch}{1.15}
\caption{Summary of Key Notation}
\label{tab:notation}
\begin{tabular}{p{0.15\columnwidth} p{0.75\columnwidth}}
\toprule
Symbol & Description \\
\midrule
\multicolumn{2}{l}{\textit{System model}} \\
$\Omega$           & Disaster area $[0,S]^2$ \\
$\mathcal{U}$, $U$ & UAV set and number of UAVs \\
$\mathcal{N}(t)$, $N$ & Active ground user set and total users \\
$\mathbf{q}_u(t)$  & 3D position of UAV $u$ at time $t$ \\
$\Delta t$         & Time slot duration \\
$P_{\max}$, $p_{n,u}$ & Max transmit power; power allocated to user $n$ by UAV $u$ \\
$g_{n,u}(t)$       & Channel power gain from UAV $u$ to user $n$ \\
$F$, $\rho_{\text{adj}}$ & Frequency reuse factor; adjacent-channel leakage \\
$B$, $B_u(t)$      & Total bandwidth; per-user bandwidth of UAV $u$ \\
$R_{n,u}(t)$       & Achievable rate for user $n$ served by UAV $u$ \\
$E_u^{\mathrm{prop}}(t)$ & Propulsion energy of UAV $u$ in slot $t$ \\
\midrule
\multicolumn{2}{l}{\textit{Phase and mobility model}} \\
$\phi_t$           & Phase index $\in\{0,1,2\}$ \\
$L$               & Phase duration (number of steps) \\
$d_n(t)$          & Demand class of user $n$: $\{L,M,H\}$ \\
$\mu_t^{\text{qoe}},\mu_t^{\text{ene}},\mu_t^{\text{col}}$ & Phase-dependent objective weights \\
$\kappa$, $\alpha$  & RPGM follow gain; GM memory coefficient \\
\midrule
\multicolumn{2}{l}{\textit{PE-MAMoE architecture}} \\
$E$, $k$          & Number of experts; top-$k$ routing parameter \\
$W_r$, $z_t^i$, $g_t^i$ & Router weights, logits, and gating probabilities \\
$\tau$             & Router temperature \\
$\theta_j$         & Parameters of expert $j$ \\
$\gamma_t^{(e)}$   & Expert noise injection scale \\
$N$               & Noise injection window (epochs after switch) \\
$\pi_{\text{ref}}$ & Frozen reference policy for KL stabilization \\
\bottomrule
\end{tabular}
\end{table}

We consider a disaster stricken square area $\Omega=[0,S]^2$, where the communication infrastructure is severely damaged, and multiple UAVs are deployed as aerial BSs to provide emergency communication coverage. Time is slotted by $t\in \{0,1,2,\dots \}$ with slot length $\Delta t$. Let $\mathcal{U}={1,\dots,U}$ be the set of UAVs. The 3D position and velocity of UAV $u\in\mathcal{U}$ at time $t$ are
\begin{align}
\mathbf{q}_u(t) &= [x_u(t), y_u(t), h_u(t)]^{\top}, \quad \forall u \in \mathcal{U}, \nonumber \\
\mathbf{v}_u(t) &= [v_{h,u}(t), v_{v,u}(t)]^{\top}.
\end{align}

The kinematics follow a discrete time update
\begin{equation}
    \mathbf{q}_u(t+1) = \mathbf{q}_u(t) + \mathbf{v}_u(t) \Delta t.
\end{equation}

subject to flight constraints
$$
0 \leq v_{h, u}(t) \leq v_{h,u}^{max}, 0 \leq v_{v, u}(t) \leq v_{v,u}^{max}.
$$
Each UAV transmits with power $P_u(t)\in[0,P_{\max}]$ (downlink).

Let $\mathcal{N}(t)={1,\dots,N_t}$ denote the active ground users (GUs) set at time $t$. The position of user $n\in\mathcal{N}(t)$ are $\mathbf{s}_{n}(t) = [x_n(t), y_n(t), z_n]^{\top}$
with motion update $\mathbf{s}_n(t{+}1)=\mathbf{s}_n(t)+\mathbf{v}_n(t)\Delta t$, where $\mathbf{v}_n(t)$ is the horizontal velocity.

\subsection{Channel Model}
\label{subsec:channle model}
For any UAV $u$ and user $n$ at time $t$, denote the horizontal distance by
$
d_{n,u}(t)=\big\|[x_n(t)-x_u(t),\,y_n(t)-y_u(t)]\big\|_2,
$
the relative height by $h_{n,u}(t)=h_u(t)-z_n$, and the $3$D distance by
$
r_{n,u}(t)=\sqrt{d_{n,u}^2(t)+h_{n,u}^2(t)}.
$

\paragraph{LoS probability}
We adopt a distance dependent and height dependent line of sight (LoS) probability
\begin{align}
P_{\mathrm{LoS}}(r,h) =
\begin{cases}
1, & r \le 18~\text{m}, \\[6pt]
\begin{aligned}
& \dfrac{18}{r} \;+\; e^{-r/36} \left(1 - \dfrac{18}{r} \right) \\
& \qquad \cdot \left(1 + C_0(h)\right)
\end{aligned}, & r > 18~\text{m}.
\end{cases}
\label{eq:plos}
\end{align}
with a height correction
$
C_0(h)=\bigl(\max\{h-13,\,0\}/10\bigr)^{1.5}.
$
We clip $P_{\mathrm{LoS}}\in[0,1]$. This follows a 3GPP UMi-style form with an additional elevation-dependent term to capture low altitude air-to-ground (A2G) links~\cite{Yan2019AccessUAVSurvey, 3gppTR38901v1610}.

\paragraph{Large scale path loss (dB)}
Conditioned on the LoS and non-LoS (NLoS) state, the path loss is modeled as
\begin{align}
&\mathrm{PL}_{\mathrm{LoS}}(r,f_c) \;=\; 28 \;+\; 22\log_{10}(r)\;+\;20\log_{10}(f_c), \label{eq:pl_los}\\
&\mathrm{PL}_{\mathrm{NLoS}}(r,f_c) \;=\; \mathrm{PL}_{\mathrm{LoS}}(r,f_c)\;+\;13\;+\;0.1\,r^{0.25},
\label{eq:pl_nlos}
\end{align}
where $r$ is in meters. The linear large scale gain is $10^{-\mathrm{PL}/10}$.\footnote{The LoS base term \eqref{eq:pl_los} matches the 3GPP UMi slope $28{+}22\log_{10}(r){+}20\log_{10}(f_c)$; the NLoS offset and weak distance exponent refinement in \eqref{eq:pl_nlos} follow our implementation to penalize NLoS links.}

\paragraph{Small scale fading and shadowing}
Small scale fading is Rician under LoS and Nakagami-$m$ under NLoS:
\begin{equation}
h_{\text{small}} \sim 
\begin{cases}
\text{Rician}(K\!\approx\!6\text{ dB}), & \text{LoS},\\
\text{Nakagami-}m~(m{=}m_{\text{NLoS}}), & \text{NLoS},
\end{cases}
\end{equation}
and multiplicative lognormal shadowing is applied with standard deviations $\sigma_{\text{LoS}}$/$\sigma_{\text{NLoS}}$ in dB. Let $G_{\mathrm{TX}},G_{\mathrm{RX}}$ be antenna gains. The per-link power gain is
\begin{equation}
g_{n,u}(t)\;=\;G_{\mathrm{TX}}G_{\mathrm{RX}}\;\bigl|\,h_{\text{small}}(t)\bigr|^2\cdot 10^{-\mathrm{PL}(r_{n,u}(t),f_c)/10}.
\label{eq:power_gain}
\end{equation}

\paragraph{Clustered reuse and interference}
We employ frequency reuse factor $F\!\in\!\mathbb{N}$ across UAVs: each UAV $u$ is assigned a color $c(u)\!\in\!\{0,\dots,F{-}1\}$.
Let $P_u(t)$, bounded by $P_{\max}$, be the total downlink power used by UAV $u$ and $p_{n,u}(t)$ the portion allocated to user $n$, $p_{n,u}(t)=0$ if $n$ is not served by $u$. The co-channel interference at user $n$ when served by UAV $u$ is
\begin{align}
I_{n,u}(t) =\;
\begin{aligned}[t]
&\sum_{\substack{u' \ne u \\ c(u') = c(u)}}\! P_{u'}(t)\,g_{n,u'}(t) \\
&\quad +\; \rho_{\text{adj}} \sum_{\substack{u' \\ c(u') \ne c(u)}}\! P_{u'}(t)\,g_{n,u'}(t),
\end{aligned}
\label{eq:interference}
\end{align}
where $\rho_{\text{adj}}\!\in\![0,1)$ models adjacent-channel leakage.

\paragraph{Bandwidth allocation, noise and SINR}
Each color receives bandwidth $B/F$. If UAV $u$ concurrently serves $K_u(t)$ users, their per user bandwidth is $B_u(t)=\frac{B}{F\,K_u(t)}$. With one sided noise spectral density $N_0$, the noise power is $N_0 B_u(t)$. The instantaneous SINR and achievable rate for user $n$ (served by $u$) are
\begin{align}
&\mathrm{SINR}_{n,u}(t) \;=\; \frac{p_{n,u}(t)\,g_{n,u}(t)}{N_0 B_u(t) + I_{n,u}(t)}, \label{eq:sinr}\\
&R_{n,u}(t) \;=\; \alpha\, B_u(t)\,\log_2\!\bigl(1+\mathrm{SINR}_{n,u}(t)\bigr), \label{eq:rate}
\end{align}
where $\alpha\!\in\!(0,1]$ captures implementation loss and coding gap. 

\paragraph{Feasibility driven power assignment}
Given a target rate $R_n^*(t)$ chosen from the demand level, the required SINR is
$
\gamma_n^*(t)=2^{\,R_n^*(t)/(\alpha B_u(t))}-1.
$
To meet $R_n^{*}(t)$ under conservative interference $I_{n,u}(t)$, the transmit power allocated to $n$ is
\begin{align}
p_{n,u}(t) &= \frac{\gamma_n^{*}(t)\,\bigl(N_0 B_u(t) + I_{n,u}(t)\bigr)}{\max\{g_{n,u}(t), \varepsilon\}},
\label{eq:power_req} \\
& \text{s.t.} \quad \sum_{n} p_{n,u}(t) \le P_{\max}
\nonumber
\end{align}
with a small $\varepsilon$ for numerical stability. Users are greedily admitted based on demand level and link efficiency until either the power budget of each UAV or the maximum number of users served by the UAV is met; finally (\ref{eq:sinr})–(\ref{eq:rate}) yield the realized SINR and rate.

\subsection{Energy Consumption Model}\label{subsec:energy_model}
Each UAV is powered by a finite onboard battery. The energy expenditure of UAV $u$ is decomposed as
\begin{equation}
    E_u(t) \;=\; E_u^{\mathrm{prop}}(t) \;+\; E_u^{\mathrm{comm}}(t).
\end{equation}

\paragraph{3D rotary wing propulsion power}
Let $V_h\!\ge 0$ and $V_z\!\in\mathbb{R}$ denote the horizontal and vertical speeds.
We adopt the rotary wing 3D propulsion power model~\cite{Yan2021NewECMRotaryWing, Zeng2019RotaryEnergyMinimization}:
\begin{align}
\label{eq:3d_power}
P_{\mathrm{rot}}(V_h,V_z)
&= P_0\!\left(1+\frac{3V_h^2}{U_{\text{tip}}^{\,2}}\right) \\
&\quad +\, P_i\!\Biggl(
\begin{aligned}[t]
&\sqrt{1+\frac{(V_h^2+V_z^2)^2}{4v_0^4}}
\\[-2pt]
&\;-\; \frac{V_h^2+V_z^2}{2v_0^2}
\end{aligned}
\Biggr)^{\!1/2} \nonumber\\
&\quad +\, \frac{1}{2}\,d_0\,\rho\,s\,A\,V_h^{3}. \nonumber
\end{align}
Here $P_0$ and $P_i$ are the profile and induced power coefficients, $U_{\text{tip}}$ is the blade tip speed, $v_0$ is the hover induced speed, and $d_0,\rho,s,A$ denote the fuselage drag coefficient, air density, rotor solidity and rotor disk area, respectively.
The hover power is $P_{\mathrm{rot}}(0,0)=P_0+P_i$.

\paragraph{Slot wise propulsion energy}
Within one slot, the controller executes horizontal and vertical movements sequentially. 
Let the realized displacements between the beginning and end of slot $t$ be $s_h(t)$ (horizontal) and $s_v(t)$ (vertical magnitude). 
Given commanded speeds $|V_h(t)|>0$ and $|V_z(t)|>0$, the motion times satisfy
\begin{align}
t_h(t) &= \frac{s_h(t)}{|V_h(t)|}, \nonumber\\[3pt]
t_v(t) &= \frac{s_v(t)}{|V_z(t)|}, \nonumber\\[3pt]
t_{\mathrm{move}}(t) &= t_h(t) + t_v(t) \;\le\; \Delta t.
\label{eq:time_relation}
\end{align}
and the remaining time is hovering $t_{\mathrm{hov}}(t)=\Delta t-t_{\mathrm{move}}(t)$.\footnote{If $t_{\mathrm{move}}(t)>\Delta t$, the action is declared infeasible in our implementation.}
Hence the propulsion energy in slot $t$ is the sum of three segments:
\begin{align}
E_u^{\mathrm{prop}}(t)
&= P_{\mathrm{rot}}\!\big(V_h(t),0\big)\,t_h(t) \\
&\quad +\, P_{\mathrm{rot}}\!\big(0,V_z(t)\big)\,t_v(t) \nonumber\\
&\quad +\, P_{\mathrm{rot}}(0,0)\,t_{\mathrm{hov}}(t).\nonumber
\label{eq:slot_prop_energy}
\end{align}

\paragraph{Battery update and budget}
Let $B_u(t)$ be the remaining battery energy at the beginning of slot $t$. The state update is
\begin{equation}\label{eq:battery_update}
B_u(t{+}1)=B_u(t)-E_u^{\mathrm{prop}}(t),\qquad 0\le B_u(t)\le B_u^{\max}.
\end{equation}
with terminal and mission feasibility requiring $B_u(t)\ge 0$ for all $t$. In this work we neglect the communication related energy on the UAV side, since extensive studies show that, for rotary UAVs, the propulsion power overwhelmingly dominates the on-board energy budget, while radio and processing energy is typically negligible in comparison; hence “UAV energy” is taken to mean flight energy throughout the paper~\cite{Wang2024RISUAV}.

\subsection{User Mobility and Phase Model}

We consider a population of $N$ GUs, indexed by $\mathcal{N}=\{1,\dots,N\}$, whose active set at time $t$ is $\mathcal{N}(t)\subseteq\mathcal{N}$. 

\paragraph{Hotspots}
Users are partitioned evenly into $G$ hotspots, $\mathcal{G}=\{1,\dots,G\}$, with fixed memberships $\{\mathcal{N}_g\}_{g=1}^G$ and $|\mathcal{N}_g|=N/G$.
Each hotspot has a center $\mathbf{c}_g(t)\in\Omega$ and a waypoint $\mathbf{p}_g(t)\in\Omega$.
Centers migrate with constant speed $u_g$ toward their targets and periodically retarget every $T_{\text{ret}}$ steps:
\begin{equation}
\mathbf{c}_g(t{+}\Delta t)
=\mathbf{c}_g(t)+u_g\,
\frac{\mathbf{p}_g(t)-\mathbf{c}_g(t)}{\|\mathbf{p}_g(t)-\mathbf{c}_g(t)\|+\varepsilon}\,\Delta t,
\label{eq:center}
\end{equation}
which follows the ``logical group center and individual perturbations" principle of the Reference Point Group Mobility (RPGM) model. Group membership is locked, ensuring each hotspot always hosts $N/G$ GUs~\cite{Camp2002MobilitySurvey,Hong1999RPGM}.

\paragraph{User mobility}
For user $n\in\mathcal{N}_g$, denote position $\mathbf{x}_n(t)\in\Omega$ and velocity $\mathbf{v}_n(t)\in\mathbb{R}^2$.
We use a Gauss Markov (GM) velocity process with memory $\alpha\in[0,1)$ and noise scale $\sigma>0$, blended with RPGM following:
\begin{equation}\label{eq:gm}
\begin{split}
\mathbf{v}_n^{\mathrm{GM}}&(t{+}\Delta t)
  = \alpha\,\mathbf{v}_n(t)
   + (1{-}\alpha)\,v_0\frac{\mathbf{v}_n(t)}{\|\mathbf{v}_n(t)\|+\varepsilon} \\
  &\quad + \sqrt{1{-}\alpha^2}\,\sigma\,\boldsymbol{\xi}_n(t),\qquad \boldsymbol{\xi}_n(t)\sim\mathcal{N}(\mathbf{0},\mathbf{I}_2)
\end{split}
\end{equation}
\begin{equation}\label{eq:rpgm}
\begin{split}
\mathbf{v}_n^{+}(t)
  &= (1{-}\kappa)\,\mathbf{v}_n^{\mathrm{GM}}(t)
   + \kappa\,v_0\frac{\mathbf{c}_g(t)-\mathbf{x}_n(t)}
                    {\|\mathbf{c}_g(t)-\mathbf{x}_n(t)\|+\varepsilon} \\
  &\quad + \boldsymbol{\eta}_n(t),\qquad \boldsymbol{\eta}_n(t)\sim\mathcal{N}(\mathbf{0},\sigma_f^2\mathbf{I}_2)
\end{split}
\end{equation}
\begin{equation}\label{eq:pos}
\mathbf{x}_n(t{+}\Delta t)
  = \Pi_{\Omega}\!\big(\mathbf{x}_n(t)+\mathbf{v}_n^{+}(t)\Delta t\big).
\end{equation}
where $\kappa\in[0,1]$ is the follow gain and $\Pi_{\Omega}(\cdot)$ denotes reflecting boundaries~\cite{Camp2002MobilitySurvey,Hong1999RPGM}.

To reflect phase driven non-stationarity, these parameters are phase dependent: in low demand phases we use slower motion with smaller drift $v_0^{\text{low}}$, noise $\sigma^{\text{low}}$ and a weaker group following coefficient $\kappa^{\text{low}}$; in medium and high demand phases we gradually increase nominal speed and directional coherence, that is, larger $v_0^{\text{med/high}}$ and $\kappa^{\text{med/high}}$, and allow stronger exploration via $\sigma^{\text{med/high}}$. 

\paragraph{Phase model}
We partition time into a fixed number of phases $\phi_t\!\in\!\{0,1,2\}$ with
equal duration $L$; after every $L$ steps the environment switches to the next
phase in a cyclic manner:
$$
\phi_t \;=\; \big\lfloor t/L \big\rfloor \bmod 3 . \nonumber
$$
The phase index $\phi_t$ is used only by the simulator to set environment parameters; it is \emph{not} included in any agent's observation. Agents must therefore infer the current regime from the observable consequences of each phase, namely the user demand distribution and spatial mobility patterns.
Each phase instantiates a distinct mobility demand regime. On the
mobility side, users follow the GM and RPGM blended kinematics with phase specific parameters
$$
v_0\!=\!v_0(\phi_t),\quad
\alpha\!=\!\alpha(\phi_t),\quad
\sigma\!=\!\sigma(\phi_t),\quad
\kappa\!=\!\kappa(\phi_t), \nonumber
$$
so that the low demand phase ($\phi\!=\!0$) uses smaller nominal speed $v_0$ and
directional noise $\sigma$, while $\phi\!=\!1,2$ progressively increase crowd
speed and directional coherence to emulate evacuation and aggregation effects in
disaster scenes. On the traffic side, each user $n$ holds a discrete demand
class $d_n(t)\!\in\!\{L,M,H\}$ whose targets and weights are phase wise
constants:
$$
R^{\mathrm{tar}}_{d}(t) = r_{d}(\phi_t),\qquad
w_{d}(t) = w_{d}(\phi_t),\qquad d\in\{L,M,H\}. \nonumber
$$
Accordingly, the per user QoS threshold used by the scheduler is
$$
R_n^{\min}(t)= r_{d_n(t)}(\phi_t),\qquad 0<r_L(\phi_0) < r_M(\phi_0) \ll r_H(\phi_0), \nonumber
$$
with $r_d(\phi)$ and $w_d(\phi)$ chosen to reflect phase priorities.

\subsection{Research Objective}
\label{subsec:research_objective}

According to the above, we cast the cooperative control of multi-UAV ECNs as a single step rolling optimization. At each decision step $t$, the controller selects UAV motion decisions so as to (i) maximize demand aware QoS utility, crediting a UAV only when a served user meets its demand class rate threshold, and (ii) minimize propulsion energy and safety penalties. QoS is treated as a soft objective via indicator terms inside the sum, which lets us reason algebraically about ``demand satisfied" and ``demand unsatisfied" events without introducing hard feasibility constraints. The UAV energy term accounts only for flight propulsion, following the common observation that communication energy is negligible compared to propulsion for rotary or fixed wing UAVs. We use the indicator bracket $[\,\cdot\,]$ to embed logical conditions into sums; it evaluates to $1$ when the statement is true and $0$ otherwise.

At each decision step $t$, the controller must jointly select the pose $q_u \in \mathbb{R}^3$ for each UAV $u$. A binary association matrix $a_{n,u}\!\in\!\{0,1\}$ indicating whether user $n\!\in\!\mathcal N$ is served by UAV $u$; and non-negative per-link bandwidth and power allocations $b_{n,u}\!\ge\!0$, $p_{n,u}\!\ge\!0$. These decisions feed the objective in \eqref{obj:system} and are constrained by the kinematics, spectrum, and power budgets mentioned below.

\newcommand{\iverson}[1]{\left[\,#1\,\right]}
\begin{align}
\label{obj:system}
\max_{q,v,a,b,p}\;\;
&\underbrace{\mu_t^{\mathrm{qoe}}
\sum_{n\in\mathcal N}\sum_{u\in\mathcal U}
a_{n,u}\; w_{d_n(t)}\!\left(
\iverson{R_{n,u}\ge r_{d_n(t)}}
\right)}_{\text{demand-aware QoS utility}}
\nonumber\\
&-\;\underbrace{\mu_t^{\mathrm{ene}}\!\sum_{u\in\mathcal U} E_u^{\mathrm{prop}}}_{\text{propulsion energy}}
\;-\;
\underbrace{\mu_t^{\mathrm{col}}\!\!\sum_{u<u'}
\iverson{\|q_u-q_{u'}\|<d_{\min}}}_{\text{safety penalty (collisions)}}
\nonumber\\
&-\;\underbrace{\lambda_{\text{ov}}\!
\sum_{u<u'}\iverson{\|q_u-q_{u'}\|<2R}\,
\exp\!\Big(-\tfrac{\|q_u-q_{u'}\|^2}{2\sigma^2}\Big)}_{\text{overlap penalty}}
\tag{O$^\prime$}
\end{align}

\noindent\textbf{Subject to:}

\begin{align}
B_u(t{+}\Delta t)= & B_u(t)-E_u^{\mathrm{prop}}(t),
\nonumber\\ & 0\le B_u(t)\le B_u^{\max}, \forall u\in\mathcal U,
\tag{C1}\label{con:batt}
\end{align}
\begin{align}
&\sum_{u\in\mathcal U} a_{n,u}\le 1,\quad a_{n,u}\in\{0,1\},\nonumber \\
&a_{n,u}=0\ \text{ if }\ \|q_u^{xy}-x_n^{xy}\|>R_{\mathrm{srv}}(\phi_t),
\forall n\in\mathcal N,\,u\in\mathcal U,
\tag{C2}\label{con:assoc}
\end{align}
\begin{align}
\sum_{n\in\mathcal N} b_{n,u}\le \frac{B}{F(\phi_t)}, &b_{n,u}\ge 0,
\sum_{n\in\mathcal N} p_{n,u}\le P_{\max}, \nonumber\\ &p_{n,u}\ge 0,
\forall u\in\mathcal U,
\tag{C3}\label{con:budgets}
\end{align}
\begin{align}
&\mathrm{SINR}_{n,u}=\frac{p_{n,u}\,g_{n,u}(q_u,x_n)}{N_0\,b_{n,u}+I_{n,u}},\nonumber\\
&R_{n,u}=\alpha\, b_{n,u}\log_2\!\bigl(1+\mathrm{SINR}_{n,u}\bigr),
\forall n\in\mathcal N,\,u\in\mathcal U,
\tag{C4}\label{con:sinr_rate}
\end{align}
\begin{equation}
\sum_{u\in\mathcal U} a_{n,u}\,R_{n,u}
\;\ge\; r_{D_n(t)}\ \mathbb{I}\!\Big(\sum_{u\in\mathcal U} a_{n,u}>0\Big),
 \forall n\in\mathcal N,
\tag{C5}\label{con:qos}
\end{equation}
\begin{equation}
\|q_u-q_{u'}\|\ge d_{\min},\qquad \forall u\neq u'\in\mathcal U.
\tag{C6}\label{con:sep}
\end{equation}

$\mathcal{U}$ and $\mathcal{N}$ are the UAV and user index sets; $q_u{=}(x_u,y_u,h_u)$ and $v_u$ denote UAV position and horizontal velocity; $a_{n,u}\!\in\!\{0,1\}$ indicates association; $b_{n,u}$ and $p_{n,u}$ are per-link bandwidth and power, bounded by \eqref{con:budgets}. The achievable rate is $R_{n,u}=\alpha\, b_{n,u}\log_2(1+\mathrm{SINR}_{n,u})$ with $\alpha\!\in\!(0,1]$ capturing implementation loss and $\mathrm{SINR}_{n,u}$ defined in \eqref{con:sinr_rate}, where $g_{n,u}(q_u,x_n)$ follows TR~38.901 large and small scale fading and $I_{n,u}$ aggregates co-channel and adjacent leakage interference. The demand class $d_n(t)\!\in\!\{L,M,H\}$ maps to a target rate $r_{d_n(t)}$ and weight $w_{d_n(t)}$; the bracket $\iverson{R_{n,u}\ge r_{d_n(t)}}$ equals $1$ if user $n$ served by UAV $u$ meets its threshold and $0$ otherwise. The energy term $E_u^{\mathrm{prop}}$ is the per-step propulsion energy of UAV $u$ with battery dynamics given in \eqref{con:batt}; communication energy is neglected relative to propulsion. The safety penalty uses a minimum separation indicator, and the overlap penalty is a truncated Gaussian with range cutoff $2R$ and scale $\sigma$ to discourage dense UAV packing while allowing flexibility inside the service radius.

Having specified \eqref{obj:system} as the task level target, the remainder of this paper addresses how to keep the policy plastic enough to track it under regime switches. \paragraph{Policy action space versus environment-side heuristics} The formulation in \eqref{obj:system} involves joint decisions over UAV poses $q_u$, user association $a_{n,u}$, bandwidth $b_{n,u}$, and power $p_{n,u}$. In our MARL instantiation, the learned policy outputs only the \emph{2D displacement} of each UAV per time slot (i.e., $a_t^i \in \mathbb{R}^2$, the horizontal movement target). All remaining resource allocation decisions are handled by deterministic, environment-side heuristics that execute within each simulator step: (i)~\emph{user association}: each user is assigned to the nearest UAV within the phase-dependent service radius $R_{\mathrm{srv}}(\phi_t)$; (ii)~\emph{user admission}: each UAV greedily admits users in decreasing demand-class priority ($H > M > L$, ties broken by proximity) up to the per-UAV capacity limit; (iii)~\emph{bandwidth allocation}: the per-color bandwidth $B/F$ is split equally among the $K_u(t)$ admitted users; (iv)~\emph{power allocation}: transmit power per user is set via the feasibility-driven rule in Eq.~\eqref{eq:power_req} to meet the demand-class rate target. This decomposition follows a common pattern in UAV-MARL research where trajectory control is learned while lower-layer radio resource management uses model-based rules, keeping the action space low-dimensional and the learning problem tractable. The policy's influence on QoS, energy, and collisions is therefore \emph{indirect}: by choosing where to fly, each UAV determines which users fall within its coverage, how link budgets are distributed, and how close it is to other UAVs.

Concretely, we design a plasticity--enhanced, sparsely gated MoE controller whose router reassigns traffic across specialized experts when the phase weights $\mu_t$ and demand combinations change, and whose experts receive light stochastic perturbations around switches to avoid common problems such as primacy bias and dormant neurons. In the next sections we instantiate this design, PE--MAMoE, analyze its stability and plasticity trade offs, and show empirically that it maintains effective rank and adaptability when the objective function undergoes sudden changes~\cite{shazeer2017outrageously,nikishin2022primacy,sokar2023dormant,dohare2024nature}.

\section{Proposed Method: PE\text{-}MAMoE Framework}
\label{sec:method}

We propose plasticity enhanced multi-agent mixture of experts, that is PE\text{-}MAMoE, a centralized training with decentralized execution (CTDE) style MARL architecture that pairs conditional computation with stochasticity based plasticity injection to keep policies adaptable when phase weights $\mu_t$ and demand combinations change abruptly. Specifically, decentralized actors consist of a sparsely gated MoE router that selects a small set of specialized experts per state, providing reconfigurable capacity without linear compute growth; a centralized critic stabilizes training. Actors are optimized with proximal policy optimization (PPO) and generalized advantage estimation (GAE) under the MAPPO recipe for cooperative games, while light, scheduled perturbations around regime switches re-plasticize the experts to avoid primacy bias and dormant units. 

\paragraph{Why UAV--ECN non-stationarity demands expert-level plasticity injection}
While stochastic perturbation is a general remedy for plasticity loss in deep RL, several properties of the phase-driven, multi-agent UAV--ECN problem make expert-level noise injection especially critical. First, UAV--ECN phase switches are \emph{abrupt and compound}: a single switch simultaneously changes user mobility parameters, demand-class weights, and QoS thresholds (Section~III-D), so the optimal joint policy shifts across multiple coupled dimensions at once; a monolithic network must overwrite its entire representation, whereas expert-level perturbation selectively loosens only the activated specialists while the router retains the regime-selection map. Second, under CTDE with parameter sharing, all UAVs execute the same actor; once an early coordination convention hardens into a few dominant experts, the team loses the ability to redistribute coverage after a mobility regime change. Expert-only noise breaks this convention lock-in at the specialist level without destabilizing the router's coarse regime assignment, a separation that has no analogue in single-agent plasticity methods. Third, emergency communication imposes hard safety constraints (collision avoidance, energy budget) alongside soft QoS objectives; global weight perturbation would jeopardize the safety-critical components learned in earlier phases, whereas confining noise to the selected experts and briefly freezing the router limits the disturbance radius and preserves the collision-avoidance policy encoded in unselected experts. Fourth, on-policy MAPPO discards data after each update, so the agent cannot revisit old-regime samples to counteract representation drift; the controlled noise window right after a switch acts as a lightweight substitute for experience replay, re-diversifying features precisely when new data is most scarce. Together, these factors make the combination of sparse MoE routing with switch-aware, expert-confined perturbation a structurally motivated design for UAV--ECN coordination rather than a generic RL enhancement.

\subsection{MoE Architecture Paradigm}
\label{subsec:moe_arch}

\paragraph{Overview}
We instantiate the decentralized actor in PE\text{-}MAMoE as a sparsely gated MoE policy: a light router selects a small subset of expert heads per agent step, and only those experts are executed; the centralized critic remains dense. This conditional computation design increases representational capacity without proportional compute, and uses a load balancing auxiliary loss to prevent expert collapse~\cite{shazeer2017outrageously,fedus2022switch}. Figure~\ref{fig:pe-mamoe-arch} sketches the information flow under CTDE: each UAV uses the shared actor with local observation~$o_t^i$, while the critic consumes the global state. Sparse routing follows top-$k$ token choice gating~\cite{shazeer2017outrageously} with switch style load balancing~\cite{fedus2022switch}.

\paragraph{Core innovations}
The PE-MAMoE framework introduces three tightly coupled innovations beyond standard sparse MoE, each targeting a distinct failure mode under phase-driven non-stationarity (see Fig.~\ref{fig:pe-mamoe-arch}).
\emph{(i)~Expert-only stochastic perturbation.} After each phase switch, controlled Gaussian noise is injected into expert parameters for a short, fixed window (Section~\ref{subsec:training_plasticity}). This re-diversifies expert features and counteracts primacy bias and dormant neurons, while leaving the router undisturbed so that the coarse regime-selection map is preserved.
\emph{(ii)~Non-parametric Phase Controller.} A set of non-gradient schedules fires at every switch: the router temperature $\tau$ is raised to encourage exploration across experts, the action log-standard-deviation is reset to re-open behavioral exploration, the entropy coefficient is reset and re-annealed, the learning rate is reduced and warmed up, and Adam moments are cleared. These schedules add zero learnable parameters yet orchestrate the transition from an ``explore new regime'' mode to a ``consolidate'' mode within each phase (Algorithm~\ref{alg:pemamoe}).
\emph{(iii)~KL-anchored stabilization.} A frozen copy of the pre-switch policy serves as a reference for a decaying KL penalty term in the total loss, preventing the policy from drifting too far too fast while the noise injection and schedules take effect. Together, these three mechanisms form a layered system that separates \emph{representational} plasticity (noise injection), \emph{behavioral} exploration (logstd reset, entropy, temperature), and \emph{stability} (KL anchor, LR warmup), enabling fast yet safe re-adaptation at every regime change.

\begin{figure}[t]
  \centering
  \includegraphics[width=0.95\linewidth]{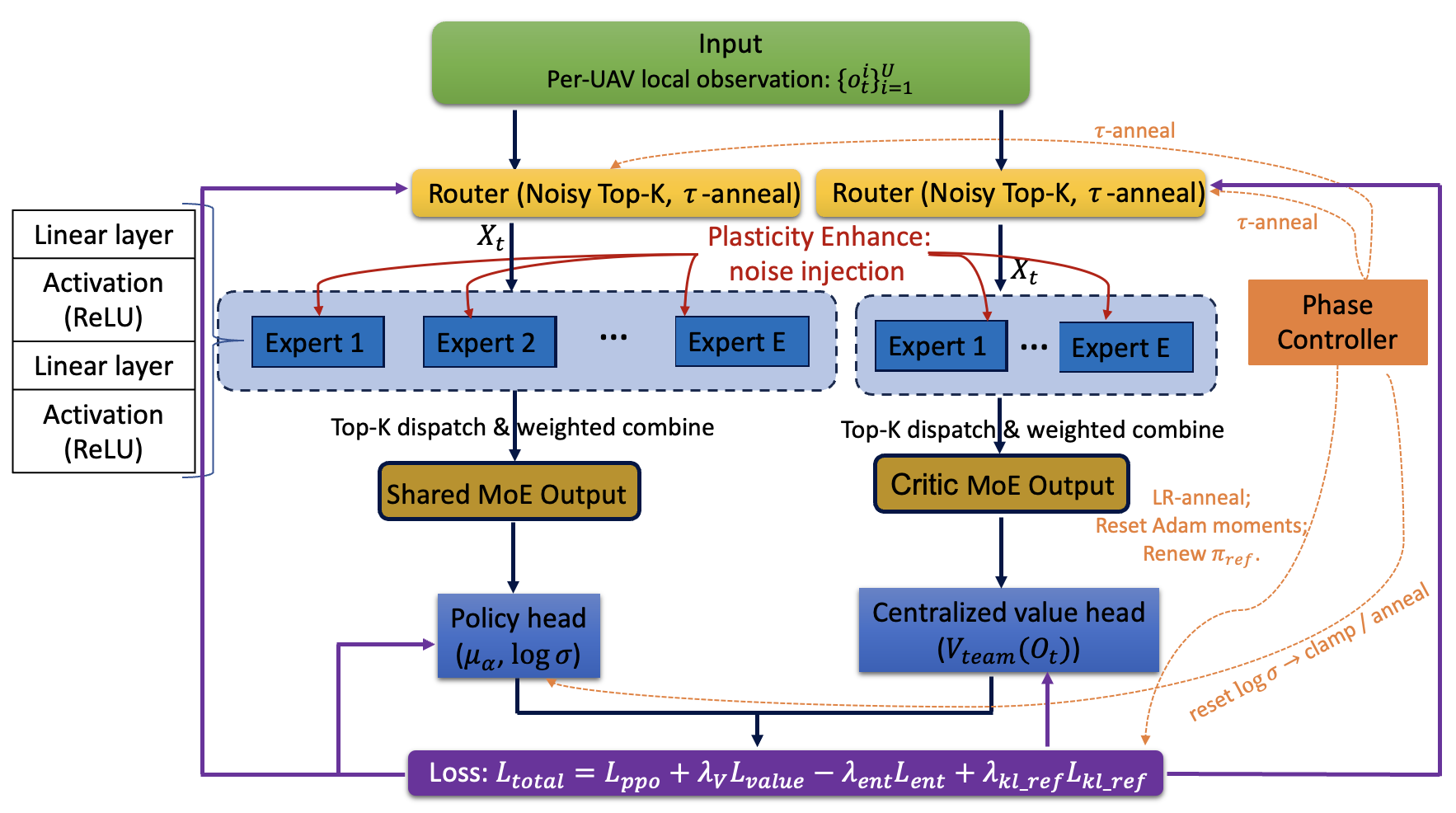}
    \caption{%
    PE-MAMoE architecture.
    For each UAV $i$, the local observation $o_t^i$ is fed into the \emph{Actor router}. The router produces gates $g$, dispatches to the Top-$k$ experts, and combines their outputs to form a shared MoE trunk, which parameterizes the policy head $(\mu_a,\log\sigma)$.
    A \emph{Phase Controller} performs non-gradient scheduling at phase switch:
    router $\tau$-anneal, per group learning rate scheduling and Adam state reset,
    and log-$\sigma$ reset $\rightarrow$ clamp/anneal; it also injects expert noise to
    enhance plasticity. On the critic side, the centralized state $X_t$ is processed by a
    \emph{Critic MoE encoder} to produce the team value $V_{\text{team}}(X_t)$.
    Training follows CTDE with the total loss
    $L_{\text{total}} = L_{\text{ppo}} + \lambda_v L_{\text{value}}
    - \lambda_{\text{ent}} L_{\text{ent}} + \lambda_{\text{kl\_ref}} L_{\text{kl\_ref}}$.
    Gradients propagate to routers and experts (conditional computation: only the selected
    experts execute).}
  \label{fig:pe-mamoe-arch}
\end{figure}

\paragraph{Router and Top-$K$ gating}
For agent $i$ at time $t$, the router input is $x_t^i = \phi(o_t^i)$, where $\phi(\cdot)$ is an identity mapping (no separate embedding) and $o_t^i$ is the local observation vector containing UAV poses, energy levels, per-user rates, per-user demand classes (normalized), and user positions. The phase preference $\mu_t$ is \emph{not} provided to the agent as a direct input; instead, the agent infers the current regime indirectly through the demand-class distribution and user spatial patterns that change at phase boundaries. The router produces logits and temperature--scaled probabilities
\begin{equation}
z_t^i=W_r x_t^i + b_r,\qquad
g_t^i=\mathrm{softmax}\!\left(\frac{z_t^i}{\tau}\right)\in\Delta^{E-1},
\end{equation}
with $E$ experts and temperature $\tau>0$. Let $\mathcal{T}_k(x_t^i)$ be indices of the $k$ largest entries of $g_t^i$ \cite{shazeer2017outrageously}. Define sparse mixture weights
\begin{equation}
w_{t,j}^i \;=\;
\begin{cases}
\dfrac{g_{t,j}^i}{\sum\limits_{m\in \mathcal{T}_k(x_t^i)} g_{t,m}^i}, & j\in \mathcal{T}_k(x_t^i),\\[6pt]
0, & \text{otherwise.}
\end{cases}
\end{equation}

\paragraph{Expert heads and fused Gaussian policy}
Each expert $f_j$ outputs an action head conditioned on $x_t^i$. For the 2D motion action with bounded displacement $a_t^i\in\mathbb{R}^2$, we use a fused Gaussian actor where expert means are linearly combined and a shared log–standard is maintained:
\begin{align}
\mu^{(j)}_t &= W^{(j)}_\mu\,x_t^i + b^{(j)}_\mu,\quad j=1,\dots,E,\\
\bar{\mu}_t^i &= \sum_{j=1}^E w_{t,j}^i\,\mu^{(j)}_t, \qquad
\log \sigma. 
\end{align}
The stochastic policy is then
\begin{equation}
\pi_\theta(a_t^i \mid s_t^i,\mu_t)=\mathcal{N}\!\left(a_t^i \ \big|\ \bar{\mu}_t^i,\ \mathrm{diag}(\sigma^2)\right),
\end{equation}
so the forward pass executes only the $k$ selected experts and the small router. 

\paragraph{Actor objective with MoE terms}
Under MAPPO, the actor minimizes
\begin{equation}
\mathcal{L}_{\mathrm{actor}}
=
\mathbb{E}\!\left[
\mathcal{L}_{\mathrm{PPO}}(\theta)
\right]
\;-\;
\lambda_{\mathrm{ent}}\,\mathbb{E}\!\left[\mathcal{H}\!\big(\pi_\theta(\cdot\mid s_t^i,\mu_t)\big)\right], \label{loss-actor}
\end{equation}
with $\lambda_{\mathrm{ent}}\!\ge\!0$. The critic and advantage estimator follow standard PPO and GAE~\cite{schulman2017ppo,schulman2016gae} under CTDE~\cite{amato2024ctde}.

\paragraph{Why MoE here}
Top-$K$ token choice routing scales capacity while keeping per-step cost near that of a small dense actor; the load balancing term keeps experts engaged and prevents idle or dominant experts, which is important for handling abrupt objective shifts~\cite{shazeer2017outrageously,fedus2022switch}. We will next describe how plasticity injection is integrated into the experts to maintain adaptability during system switching.

\subsection{Training Experts with Plasticity Injection}
\label{subsec:training_plasticity}

\paragraph{Design Motivation}
Deep RL policies gradually lose the ability to fit new targets after long training, a phenomenon known as loss of plasticity, manifested by dormant neurons and rank collapse; this is aggravated under non-stationary objectives and on policy updates~\cite{nikishin2022primacy,sokar2023dormant,lyle2024plasticity,dohare2024nature}.%
\footnote{Empirical evidence spans Atari/MuJoCo and classification; maintaining plasticity typically requires continual diversity injection or resets.}
To keep MoE experts adaptable when phase weights $\mu_t$ and demand combinations change, PE\text{-}MAMoE couples light stochastic perturbations with targeted reactivation at regime switches, while retaining PPO and GAE optimization. As discussed above, expert-only injection is chosen over global perturbation because UAV--ECN switches are compound (mobility, demand, and QoS shift simultaneously) and safety critical (collision avoidance must persist); confining noise to the activated specialists re-plasticizes the capacity needed for the new regime while the frozen router and unselected experts protect previously learned safe behaviors.

\paragraph{Where the noise enters}
Let $\theta_j$ denote the parameters of expert $j\!\in\!\{1,\dots,E\}$. For a minibatch $\mathcal{B}$, the PPO actor loss with MoE terms (see Eq.~\eqref{loss-actor}) is optimized with noise-perturbed parameter updates on the selected experts:
\begin{align}
\theta_j &\leftarrow \theta_j - \eta_t \nabla_{\theta_j}\mathcal{L}_{\mathrm{actor}}
\;+\; \eta_t\,\gamma_t^{(e)}\,\varepsilon_j,
\varepsilon_j\sim\mathcal{N}(\mathbf{0},\mathbf{I}),\ \forall j\in\mathcal{T}_k. \label{eq:expert-noise}
\end{align}
where $\eta_t$ is the learning rate, and $\gamma_t^{(e)}$ is switch-aware noise scales.

\paragraph{Switch aware noise window}
Let $t_s$ be a phase switch time. We use a fixed amplitude pulse for $N$ epochs after each switch and zero otherwise:

\begin{align}
\gamma_t^{(e)}=
\begin{cases}
\gamma_0^{(e)}, & t\in[t_s,\,t_s{+}N),\\
0, & \text{otherwise}.
\end{cases}
\end{align}

In the main experiments we set $\gamma_0{=}0.005$. In the ablation (Sec.~\ref{subsec:ablation_study}) we set $\gamma_0$ with different values and apply it only within the first $N$ epochs after a switch.

\paragraph{Monitoring dormant units}
To quantify plasticity degradation, we only monitor the dormant neuron fraction (DNF) per layer during training and across phase switches, without applying any reinitialization. For a layer of width $d$, the DNF over a minibatch $\mathcal{B}$ is
\begin{equation}
\mathrm{DNF} \;=\; \frac{1}{d}\sum_{u=1}^d 
\mathbb{I}\!\left\{
\frac{1}{|\mathcal{B}|}\sum_{x\in\mathcal{B}} |a_u(x)| \;<\; \delta
\right\},
\end{equation}
where $a_u(\cdot)$ is the output before nonlinear activation and $\delta$ is a small threshold. We report DNF together with effective rank statistics to diagnose plasticity loss and representation collapse under non-stationary phases.

\paragraph{Router temperature control}
To avoid early router collapse while keeping training simple, we do not add a router entropy regularizer to the loss. Instead, a non-gradient Phase Controller adjusts the router temperature immediately after a phase switch and then cools it as the phase stabilizes:
\begin{equation}
g_t=\mathrm{softmax}\!\Big(\frac{z_t}{\tau_t}\Big),\qquad
\tau_t=\tau_{\min}+\frac{\tau_{\max}-\tau_{\min}}{1+\kappa\,e^{-\beta\,\Delta(t)}},
\label{eq:router-temp-only}
\end{equation}
where $\Delta(t)$ is time since the last switch. This temperature schedule increases routing diversity
right after switches without introducing an extra loss term.

\paragraph{Stability}
Training follows CTDE with PPO and GAE. The total loss used in our experiments is
\begin{align}
\min_{\theta,\psi}\ \mathcal{L}_{\mathrm{total}}
&=\underbrace{\mathbb{E}_{\mathcal{B}}\!\big[\mathcal{L}_{\mathrm{PPO}}(\theta)\big]}_{\text{actor}}
\;+\;\underbrace{\lambda_V\,\mathbb{E}_{\mathcal{B}}\!\big[(V_\psi-\hat V)^2\big]}_{\text{value}} \\
& \;-\;\lambda_{\mathrm{ent}}\,\mathbb{E}_{\mathcal{B}}\!\big[\mathcal{H}(\pi_\theta)\big]
\;+\;\lambda_{\mathrm{kl\_ref}}\,\mathbb{E}_{\mathcal{B}}\!\big[\mathrm{KL}(\pi_\theta\,\|\,\pi_{\text{ref}})\big],\nonumber
\label{eq:total-loss-final}
\end{align}
where $\pi_{\text{ref}}$ is a reference policy used for stabilization across switches, set to the pre-switch actor snapshot. Parameters are updated with the
expert noise rule in Eq.~\eqref{eq:expert-noise}. The Phase Controller also applies non-gradient schedules at switches:
per-group learning rate warmup and decay, Adam state reset, and actor log-$\sigma$ reset followed by
clamp and anneal (cf.\ Fig.~\ref{fig:pe-mamoe-arch}).

Algorithm~\ref{alg:pemamoe} consolidates the full PE-MAMoE training procedure, making the execution order and conditional logic explicit.

\begin{algorithm}[!t]
\caption{PE-MAMoE Training Loop}
\label{alg:pemamoe}
\begin{algorithmic}[1]
\footnotesize
\STATE \textbf{Input:} phase schedule $\{\phi_0,\phi_1,\dots\}$, phase length $L$, experts $\{f_j\}_{j=1}^E$, router $W_r$, noise budget $N$, noise scale $\gamma_0$, temperature bounds $\tau_{\min},\tau_{\max}$, freeze window $K_{\text{freeze}}$, warm window $K_{\text{warm}}$
\STATE Initialize shared actor $\pi_\theta$ (MoE), centralized critic $V_\psi$, optimizer (Adam)
\STATE $\text{noise\_left} \leftarrow 0$; $\text{freeze\_left} \leftarrow 0$; $\pi_{\text{ref}} \leftarrow \text{None}$
\FOR{iteration $= 1, 2, \dots$}
    \STATE \textbf{// --- Check phase switch ---}
    \IF{iteration $\bmod$ $L = 0$}
        \STATE $\phi \leftarrow$ next phase; update env demands and mobility
        \STATE \textbf{// Phase Controller (non-gradient):}
        \STATE $\pi_{\text{ref}} \leftarrow$ frozen copy of $\pi_\theta$ \hfill $\triangleright$ KL anchor
        \STATE Reset Adam optimizer states (moments $\leftarrow 0$)
        \STATE Inject expert noise: $\theta_j \leftarrow \theta_j + \gamma_0 \varepsilon_j,\ \varepsilon_j\!\sim\!\mathcal{N}(\mathbf{0},\mathbf{I}),\ \forall j$
        \STATE $\text{noise\_left} \leftarrow N$; \ $\text{freeze\_left} \leftarrow K_{\text{freeze}}$
        \STATE Reset $\log\sigma \leftarrow \sigma_{\text{init}}$ \hfill $\triangleright$ re-open exploration
        \STATE Reset entropy coef.\ to initial value; begin anneal
        \STATE Set LR $\leftarrow \alpha_{\text{phase}} \cdot \alpha_0$; begin warmup over $K_{\text{warm}}$ iters
        \STATE Set conservative PPO clip and target-KL for $K_{\text{warm}}$ iters
    \ENDIF
    \STATE \textbf{// --- Scheduling within phase ---}
    \IF{freeze\_left $> 0$}
        \STATE Set router LR $\leftarrow 0$ (freeze router weights)
        \STATE freeze\_left $\leftarrow$ freeze\_left $- 1$
    \ELSE
        \STATE Restore router LR; warm up toward $\alpha_0$
    \ENDIF
    \STATE Anneal $\tau_t$ via Eq.~\eqref{eq:router-temp-only} based on time since last switch
    \STATE Anneal entropy coefficient linearly within current phase
    \STATE \textbf{// --- Rollout ---}
    \STATE Collect trajectories $\{(o_t^i, a_t^i, r_t, o_{t+1}^i)\}$ using $\pi_\theta$ with top-$k$ routing
    \STATE Compute advantages $\hat{A}_t$ via GAE with centralized critic $V_\psi$
    \STATE \textbf{// --- PPO update epochs ---}
    \FOR{epoch $= 1, \dots, E_{\text{update}}$}
        \FOR{each minibatch $\mathcal{B}$}
            \STATE Compute $\mathcal{L}_{\text{PPO}}$, $\mathcal{L}_{\text{value}}$, $\mathcal{H}(\pi_\theta)$
            \IF{$\pi_{\text{ref}} \neq \text{None}$ \AND $\lambda_{\text{kl\_ref}} > 0$}
                \STATE $\mathcal{L}_{\text{total}} \leftarrow \mathcal{L}_{\text{PPO}} + \lambda_V \mathcal{L}_{\text{value}} - \lambda_{\text{ent}} \mathcal{H} + \lambda_{\text{kl\_ref}} \mathrm{KL}(\pi_\theta \| \pi_{\text{ref}})$
            \ELSE
                \STATE $\mathcal{L}_{\text{total}} \leftarrow \mathcal{L}_{\text{PPO}} + \lambda_V \mathcal{L}_{\text{value}} - \lambda_{\text{ent}} \mathcal{H}$
            \ENDIF
            \STATE Update $\theta, \psi$ via Adam on $\mathcal{L}_{\text{total}}$
        \ENDFOR
        \IF{noise\_left $> 0$}
            \STATE Inject expert noise: $\theta_j \leftarrow \theta_j + \gamma_0 \varepsilon_j,\ \forall j \in \mathcal{T}_k$
            \STATE noise\_left $\leftarrow$ noise\_left $- 1$
        \ENDIF
    \ENDFOR
    \STATE Decay $\lambda_{\text{kl\_ref}}$; if below threshold, release $\pi_{\text{ref}}$
\ENDFOR
\end{algorithmic}
\end{algorithm}

\section{Theoretical Analysis}
\label{sec:theory}

We analyze PE\text{-}MAMoE under phase driven non-stationarity and prove a dynamic regret bound showing that the tracking error is governed by two factors: the variation budget of the changing environment and the cumulative noise energy injected to restore plasticity.

\paragraph{Problem setting}
Let $\{\mathcal M_t\}_{t=1}^T$ be a sequence of cooperative Markov decision processes (MDPs), induced by phase switches. At step $t$, let $\pi_t^*$ be the optimal policy for $\mathcal M_t$, and let $J_t(\pi)$ denote its episodic performance criterion. The dynamic regret of a learning policy sequence $\{\pi_t\}$ is
\begin{equation}
\mathrm{DynRegret}_T \;\triangleq\; \sum_{t=1}^T \Big( J_t(\pi_t^*) - J_t(\pi_t) \Big).
\end{equation}
We quantify non-stationarity by a variation budget $V_T$ (a.k.a.\ path length), instantiated either on rewards and transitions $(B_r,B_p)$ or on the path length of the optimal values and policies, as standard in non-stationary RL~\cite{cheung2020nonstationary,fei2020dynamic}.

\subsection{Assumptions}
\label{subsec:assumptions}
\begin{enumerate}
\item[(A1)] \textbf{Bounded variation.} The environment drift is controlled: $V_T<\infty$, e.g., bounded total variation in rewards and transitions~\cite{cheung2020nonstationary}, or bounded path length in $J_t(\pi_t^*)$.
\item[(A2)] \textbf{PPO and trust region policy optimization (TRPO)-style stability.} Per-step policy updates obey a TRPO and PPO performance difference bound with a trust region surrogate~\cite{wang2019trpo_guided,ppo_fisher_rao}.
\item[(A3)] \textbf{Centralized critic error bound.} The CTDE critic yields bounded estimation error; GAE reduces variance, producing a summable martingale remainder.
\item[(A4)] \textbf{MoE Lipschitz stability.}
Assume router logits and expert policies are Lipschitz in parameters and the
Top-$k$ selection uses a bounded temperature $\tau_t$, yielding a piecewise smooth
mixed policy $\pi_\theta(a\mid s,\phi)=\sum_{k=1}^K G_k(s,\phi)\,\pi_{\theta_k}(a\mid s)$ with
\begin{align}
\|\pi_{\theta'}(\cdot\mid s,\phi)-\pi_\theta(\cdot\mid s,\phi)\|_1
\;\le\; & L_G\|G_{\theta'}-G_\theta\| \; \\
& + \; \sum_{k=1}^K L_k\|\theta'_k-\theta_k\|. \nonumber
\end{align}
Thus MoE does not violate TRPO and PPO stability; it only affects constants in the
performance difference bound~\cite{chen2022moe_theory,baykal2022sparse,zhou2022expertchoice}.
\item[(A5)] \textbf{Expert only noise.}
At training step $t$, we inject zero mean Gaussian perturbations only into the
selected experts.
The expert noise scale $\gamma_t^{(e)}$ is a constant and a fixed amplitude pulse restricted to the first $N$ epochs after each phase switch. Hence the cumulative ``noise energy” $\sum_{t=1}^T (\gamma_t^{(e)})^2$ is finite.
\end{enumerate}

\subsection{Auxiliary lemmas}

\begin{lemma}[TRPO or PPO performance difference]
\label{lem:ppo_stability}
Let $\pi'$ be the PPO or TRPO update of $\pi$ under a trust region surrogate. Then for the MDP at time $t$,
\begin{equation}
J_t(\pi') - J_t(\pi) \;\ge\; \mathbb{E}_{s\sim d_t^\pi,a\sim\pi'}\!\big[ A_t^\pi(s,a) \big] \!-\! C_t\,\mathrm{Div}(\pi'\Vert \pi) \!-\! \varepsilon_t^{\mathrm{est}},
\end{equation}
where $A_t^\pi$ is the advantage under $\pi$ in $\mathcal M_t$, $\mathrm{Div}$ is a divergence, $C_t$ is a constant depending on mixing and reward bounds, and $\varepsilon_t^{\mathrm{est}}$ collects critic or GAE approximation errors~\cite{wang2019trpo_guided,ppo_fisher_rao}.
\end{lemma}

\begin{lemma}[Drift decomposition under variation budget]
\label{lem:drift_variation}
Let $\{\pi_t\},\{\pi_t^*\}$ be the learner and step wise optimal policies. Then
\begin{align}
\sum_{t=1}^T \!\big( J_t(\pi_t^*) - J_t(\pi_t) \big)
\;\le\; \underbrace{\sum_{t=1}^T \!\big( J_t(\pi_t^*) - J_{t-1}(\pi_{t-1}^*) \big)}_{\textstyle \mathcal{D}_T(V_T)}+ \\
\underbrace{\sum_{t=1}^T \!\big( J_t(\pi_{t}) - J_t(\pi_{t-1}) \big)^{-}}_{\textstyle \mathcal{O}_T}
\;+\; \sum_{t=1}^T \varepsilon_t^{\mathrm{est}},\nonumber
\end{align}
where $\mathcal{D}_T(V_T)$ is controlled by the environment variation budget $V_T$, and $\mathcal{O}_T$ is the optimization cost bounded via Lemma~\ref{lem:ppo_stability} by trust region terms. Standard non-stationary RL arguments further give a representative bound $\mathcal{D}_T(V_T) \!\in\! \widetilde{\mathcal O}(T^{3/4} V_T^{1/4})$ under common mixing or diameter conditions~\cite{cheung2020nonstationary,fei2020dynamic}.
\end{lemma}

\begin{lemma}[Noise energy penalty]
\label{lem:noise_energy}
Under \textnormal{(A5)}, noise perturbed updates add a martingale difference
variance term so that
\begin{equation}
    \sum_{t=1}^T \big(J_t(\pi_t)-\mathbb{E}[J_t(\pi_t)\mid\mathcal{F}_{t-1}]\big)
\;=\; \mathcal{O}_{\mathbb{P}}\!\Big(\sqrt{\textstyle\sum_{t=1}^T (\gamma_t^{(e)})^2}\Big).
\end{equation}
Consequently, the cumulative penalty scales as $\mathcal{O}(\gamma\sqrt{T})$ for
constant $\gamma$, and as $\mathcal{O}(\gamma\sqrt{K N})$ when a fixed amplitude
pulse is applied for $N$ epochs after each of $K$ phase switches, independent of
annealing.
\end{lemma}

\begin{lemma}[MoE routing stability]
\label{lem:moe_stability}
Under (A4), the mixed policy $\pi_\theta(a\!\mid\!s,\phi) = \sum_{k=1}^K G_k(s,\phi)\,\pi_{\theta_k}(a\!\mid\!s)$ satisfies
\begin{equation}
\|\pi_{\theta'}(\cdot\!\mid\!s,\phi) - \pi_\theta(\cdot\!\mid\!s,\phi)\|_1
\!\le\! L_G \|G_{\theta'}-G_\theta\| \!+\! \sum_{k=1}^K L_k \|\theta_k'\!-\!\theta_k\|,
\end{equation}
with Lipschitz constants $(L_G,L_k)$ that depend on Top-$K$ sparsity and load balancing or entropy regularization. Hence MoE does not violate the TRPO or PPO stability preconditions; it only affects constants in Lemma~\ref{lem:ppo_stability}.~\cite{chen2022moe_theory,baykal2022sparse,zhou2022expertchoice}.
\end{lemma}

\subsection{Main result}

\begin{theorem}[Dynamic regret of PE\text{-}MAMoE]
\label{thm:main}
Under Assumptions~\ref{subsec:assumptions} and Lemmas~\ref{lem:ppo_stability}--\ref{lem:moe_stability}, the dynamic regret of PE\text{-}MAMoE over $T$ steps admits
\begin{align}
\mathrm{DynRegret}_T(\text{PE\text{-}MAMoE}) \le &
\widetilde{\mathcal O}\!\big(T^{3/4} V_T^{1/4}\big)
+ \\
& \mathcal O\!\Big(\sqrt{\textstyle\sum_{t=1}^T \gamma_t^2}\Big) +\mathfrak{E}_{\mathrm{est}}(T), \nonumber
\end{align}
where $\mathfrak{E}_{\mathrm{est}}(T)$ aggregates critic or GAE approximation terms.
\end{theorem}

\begin{proof}
By Lemma~\ref{lem:drift_variation},
\begin{equation}
\mathrm{DynRegret}_T \le \mathcal{D}_T(V_T) + \mathcal{O}_T + \sum_t \varepsilon_t^{\mathrm{est}}.
\end{equation}
The drift term $\mathcal{D}_T(V_T)$ yields the $\widetilde{\mathcal O}(T^{3/4} V_T^{1/4})$ dependence under standard non-stationary MDP arguments~\cite{cheung2020nonstationary,fei2020dynamic}. For $\mathcal{O}_T$, Lemma~\ref{lem:ppo_stability} upper bounds the optimization cost by trust region remainders; MoE routing respects the stability preconditions by Lemma~\ref{lem:moe_stability}. The additional penalty from expert perturbations scales as $\mathcal O(\sqrt{\sum_t \gamma_t^2})$ by Lemma~\ref{lem:noise_energy}. Summing all contributions and grouping estimation terms into $\mathfrak{E}_{\mathrm{est}}(T)$ proves the claim.
\end{proof}

\section{Results and Analysis}

This section evaluates PE-MAMoE through a progression from setup and internal diagnostics to system level outcomes and sensitivity analysis. Section~\ref{subsec:sim_env} describes the phase driven simulation environment and its parameters. Section~\ref{subsec:baselines} introduces the three baselines that share the same MAPPO pipeline but lack plasticity mechanisms. Section~\ref{subsec:effectiveness} examines effectiveness from two complementary angles: plasticity indicators (expert feature rank and dormant neuron fraction) that reveal whether the representation stays expressive across regime switches, and task performance metrics (return and normalized IQM) that quantify the end result. Section~\ref{subsec:expert_selection} inspects how the policy and value routers redistribute expert traffic at phase boundaries, linking routing diversity to the plasticity and load balancing design. Section~\ref{subsec:system_performance} reports system level metrics, including collision rate, energy consumption, served user count, and UAV utilization, to assess operational robustness. Section~\ref{subsec:internal_states} visualizes action level responses to regime switches, connecting internal control behavior to the system outcomes observed earlier. Finally, Section~\ref{subsec:ablation_study} isolates individual design components via ablation and sweeps the noise injection threshold to characterize the sensitivity of plasticity injection.

\subsection{Simulation Environment}
\label{subsec:sim_env}

We simulate a $1~\mathrm{km}\!\times\!1~\mathrm{km}$ post disaster urban square with slotted time $\Delta t$. Non-stationarity is driven by a phase variable $\phi_t\!\in\!\{0,1,2\}$ that jointly switches user mobility regime and demand mix, as well as radio reuse ($F$) and UAV service radius $R_{\mathrm{srv}}(\phi_t)$. Users follow phase dependent crowd like mobility: as the system transitions from low to medium and high demand phases, both the nominal speed and directional persistence increase, which induces abrupt density shifts and hotspot formation at the switch times. The simulation parameters are summarized in Table~\ref{tab:sim-config}.

\begin{table}[t]
\centering
\scriptsize
\setlength{\tabcolsep}{3pt}
\caption{Simulation and training parameters.}
\label{tab:sim-config}

\begin{adjustbox}{width=0.92\columnwidth,center} 
\begin{tabularx}{\linewidth}{
  @{} >{\raggedright\arraybackslash}p{.47\linewidth}
      >{\raggedright\arraybackslash}X @{}}
\toprule
Area & $1000 \times 1000$ (m) \\
Number of UAVs & $3$ \\
Concurrent users & $20$ \\
Time step & $60$ s \\
Episode length & $32$ steps \\
Service radius & $\{200, 150, 150\}$ (m) \\
Max users per UAV & $5$ \\
Reuse factor & $1$ \\
Adjacent leakage & $10^{-3}$ \\
\midrule
\multicolumn{2}{@{}l}{\textbf{Mobility and Demand Phases}}\\
\midrule
Phase model & phases $[0,1,2]$ \\
Repeat phases & $3$ cycles \\
User moves & \texttt{True} \\
User speed & $\{0.2, 0.5, 0.8\}$ (m/s) \\
Direction jitter& $\{0.2, 0.5, 0.8\}$ \\
Mobility model & GM  \\
GM memory & $0.9$ \\
GM noise & $0.08$ (m/s) \\
RPGM & enabled; groups $=3$, follow gain $=0.3$ \\
RPGM follow noise & $0.15$ (m/s) \\
RPGM center speed & $1.2\times$ \texttt{User\_speed} \\
Demand mode & \texttt{cycle}; cycle length $=3$ (L$\to$M$\to$H) \\
Demand thresholds (Mbps) & L: $0.5$, M: $1.0$, H: $2.0$ \\
\midrule
\multicolumn{2}{@{}l}{\textbf{Reward / Penalties}}\\
\midrule
Success cover reward (a user) $\mu_t^{\text{qoe}}$ & L: $5.0$, M: $10.0$, H: $20.0$ \\
Energy penalty $\mu_t^{\text{ene}}$ & $0.1$ \\
Collision penalty $\mu_t^{\text{col}}$ & $100.0$ \\
Overlap penalty $\lambda_{\text{ov}}$ & $\{20, 40, 80\}$ \\
Overlap scale $\sigma$ & $150$ m \\
\midrule
\multicolumn{2}{@{}l}{\textbf{Policy / Value Network}}\\
\midrule
Actor/Value type & \texttt{mlp, moe, smoe, PEMAMoE} \\
Activation & ReLU \\
Experts per MoE & $3$ \\
Top-$k$ routing & $k=1$ \\
Expert hidden size & $32$ \\
\midrule
\multicolumn{2}{@{}l}{\textbf{PPO / MAPPO Training}}\\
\midrule
Total timesteps & $22{,}118{,}400$ \\
Rollout length & $2048$ \\
Minibatches & $32$ \\
Update epochs & $8$ \\
Learning rate & $3\times 10^{-4}$ \\
Weight decay & $10^{-4}$ \\
Adam $(\beta_1,\beta_2)$ & $(0.9,\,0.999)$ \\
Discount $\gamma$ & $0.99$ \\
GAE $\lambda$ & $0.95$ \\
Clip coef. & $0.15$ \\
Value loss coef. & $2.0$ \\
Entropy coef.\ start/end & $0.01 \rightarrow 0.003$ \\
Entropy anneal portion & $0.3$ (of training) \\
Anneal LR & \texttt{False} \\
Target KL & $0.1$ \\
Max grad norm & $0.5$ \\
\midrule
\multicolumn{2}{@{}l}{\textbf{Switch-time Schedules}}\\
\midrule
KL-to-ref (init / decay) & $0.1$ / $0.985$ per iteration \\
Warm window after switch & $K_{\text{warm}}=80$ iters \\
Freeze router window & $K_{\text{freeze\_router}}=30$ iters \\
Phase LR multiplier & $0.3$ (warmup back over $200$ iters) \\
Log-std reset & $0.3$ \\
Noise Injection epochs & $30$ \\
\bottomrule
\end{tabularx}
\end{adjustbox}
\end{table}

\subsection{Baselines}
\label{subsec:baselines}

We benchmark PE\text{-}MAMoE against three non-plasticity baselines implemented under the same CTDE and MAPPO pipeline and environment dynamics:

\paragraph{MLP.}
A shared parameter multilayer perceptron without conditional computation. The actor and critic are standard feedforward networks, trained with PPO and GAE; no expert routing, load balancing, or stochasticity injection is used. This represents a strong vanilla baseline commonly adopted in MADRL. 

\paragraph{MoE.}
A mixture of experts policy where a learned router mixes multiple expert sub-networks to produce the action distribution, but without plasticity injection, router temperature and entropy scheduling. Experts are trained jointly via the PPO loss; the router outputs a combination over all experts each step. This baseline isolates the benefit of conditional computation from plasticity maintenance. 

\paragraph{Sparse MoE.}
A sparse MoE variant that routes each input to the top-$k$ experts and combines only their outputs; no plasticity injection is applied. This reduces communication versus dense mixing and is widely used in large scale sparse models. Our SMoE matches the MoE capacity while enforcing sparse activation through top-$k$ gating. 

\medskip
\noindent\textbf{Implementation notes.}
All baselines share the same observation and action interfaces, PPO hyperparameters, and expert counts and hidden sizes as PE\text{-}MAMoE for fairness. The MoE and SMoE baselines differ only in their routing rule and exclude any plasticity preserving mechanisms, thereby attributing gains to conditional computation alone rather than re-plasticization.

\paragraph{Baseline selection rationale}
Our baseline set is designed to perform a controlled \emph{architectural} ablation within a single, consistent MAPPO pipeline: MLP tests whether a monolithic network suffices, dense MoE tests whether conditional computation alone helps, and Sparse MoE tests whether sparsity improves over dense mixing. PE-MAMoE adds plasticity injection on top of Sparse MoE, so the progression MLP $\to$ MoE $\to$ SMoE $\to$ PE-MAMoE isolates each design factor. We acknowledge that the current set does not include methods from the non-stationary RL or continual learning literature, such as periodic full-network resets~\cite{nikishin2022primacy}, elastic weight consolidation (EWC)~\cite{kirkpatrick2017ewc}, policy distillation, or meta-RL approaches~\cite{finn2017maml}. These methods originate from different problem formulations (single-agent, supervised, or task-family settings) and would require non-trivial adaptation to our cooperative MAPPO pipeline with shared parameters and phase-driven objective reweighting. In particular, full-network reset destroys all inter-agent coordination conventions; EWC requires per-task Fisher information matrices that are ill-defined when tasks blend continuously; and meta-RL presupposes a task distribution for meta-training, which is unavailable in our single-environment, phase-cycling setup. We regard comparison with suitably adapted versions of these methods as valuable future work.

\paragraph{Environment generality}
Our UAV--ECN simulator is released as open-source and exposes a standard Gymnasium-compatible multi-agent interface (observation dicts, continuous action spaces, step/reset protocol). Any MARL algorithm that accepts this interface can be evaluated without modification to the environment code. The phase-driven demand and mobility switching mechanism is a configurable module: phase count, duration, cycling order, demand thresholds, and mobility parameters are all exposed as configuration variables (Table~\ref{tab:sim-config}). We deliberately chose this design so that the environment is not tailored to PE-MAMoE; the same simulator equally challenges any policy, as evidenced by the fact that all four methods, including the plain MLP baseline, were trained and evaluated under identical dynamics with no environment-side adjustments.

\paragraph{Computational cost}
Table~\ref{tab:cost} compares the parameter count and per-step inference cost of each method. With an observation dimension of $177$, expert hidden size $32$, and $E{=}3$ experts, the MLP actor contains roughly $6.8$k parameters. MoE and SMoE triple the expert parameters and add a router, reaching about $20.9$k and $21.4$k, respectively. PE-MAMoE has the same parametric footprint as SMoE ($21.4$k) because its Phase Controller (temperature annealing, learning rate warmup, log-$\sigma$ reset, and noise injection) operates entirely through non-gradient schedules and adds no learnable parameters. At inference, top-$1$ routing activates only one expert plus the router per step, so the effective multiply-accumulate cost is approximately $1.16\times$ that of MLP despite the $3.1\times$ parameter ratio, retaining the conditional computation advantage of sparse MoE. Training wall-clock time increases by roughly 15--20\% over MLP, predominantly from the router forward pass and the per-switch noise injection window ($N{=}30$ epochs out of $\sim$10\,800 total).

\begin{table}[!t]
\centering
\scriptsize
\renewcommand{\arraystretch}{1.2}
\caption{Computational cost comparison (actor network).}
\label{tab:cost}
\begin{tabular}{lccc}
\toprule
Method & Actor params & Params vs.\ MLP & Inference active params \\
\midrule
MLP        & 6\,820   & 1.00$\times$  & 6\,820 (all) \\
MoE        & 20\,858  & 3.06$\times$  & 20\,858 (all experts) \\
SMoE       & 21\,392  & 3.14$\times$  & 7\,820 (top-1 + router) \\
PE-MAMoE   & 21\,392  & 3.14$\times$  & 7\,820 (top-1 + router) \\
\bottomrule
\end{tabular}
\end{table}

\subsection{Effectiveness of PE\text{-}MAMoE}
\label{subsec:effectiveness}

\subsubsection{Plasticity indicators: feature rank and dormant neurons}
Fig.~\ref{fig:rank} reports the average feature rank across experts over training.%
\footnote{We use the stable rank of expert activations as a proxy for representational expressiveness; rank collapse is associated with loss of plasticity~\cite{lyle2023understanding}.}
PE\text{-}MAMoE consistently maintains the highest rank among all methods and exhibits a clear phase-synchronous, periodic pattern: at each environment switch, the rank briefly decreases and then recovers to a high plateau before the next switch. This cyclical recovery evidences rapid re-plasticization followed by re-stabilization under the new objective, rather than settling into a low-rank, brittle representation. In contrast, vanilla MoE quickly drifts toward a low, nearly stationary rank, and SMoE shows step wise degradations that only partially recover. 

\begin{figure}[t]
    \centering
    \includegraphics[width=\linewidth]{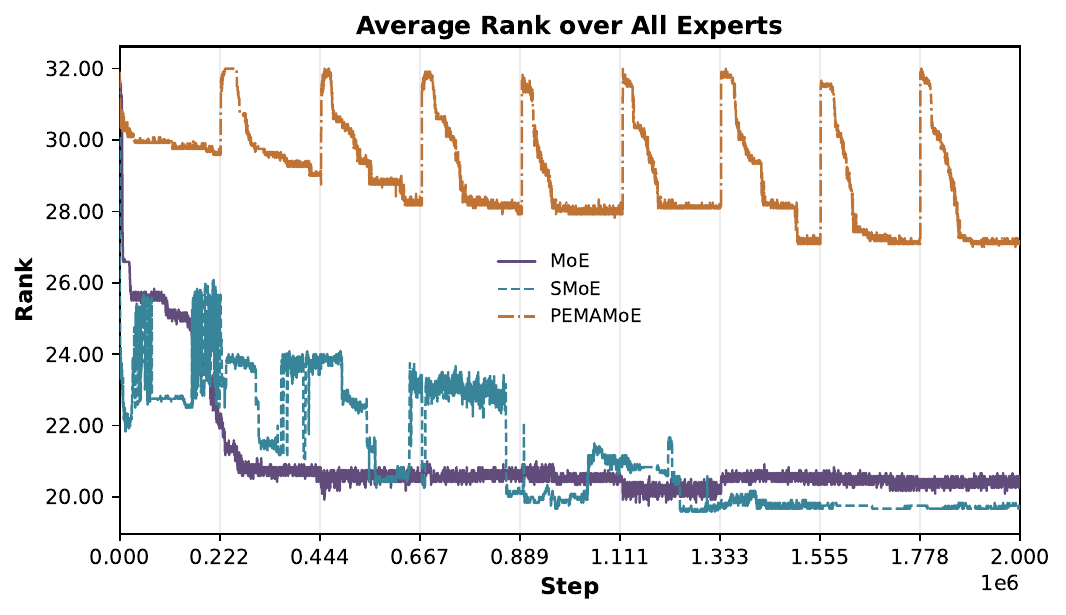}
    \caption{Evolution of the average expert feature rank over training.}
    \label{fig:rank}
\end{figure}

Fig.~\ref{fig:dormant} complements this view using the dormant neuron fraction. While PE\text{-}MAMoE’s DNF can be temporarily higher than SMoE within a phase, it exhibits periodic troughs aligned with regime switches, indicating reactivation of previously quiescent channels. This is precisely the behavior desired under non-stationarity: neurons are recruited when objectives change and then pruned back as the policy re-stabilizes. By contrast, vanilla MoE accumulates a large, persistent dormant set, consistent with the ``dormant neuron phenomenon" that harms expressivity in deep RL~\cite{sokar2023dormant}. The periodic reactivation in PE\text{-}MAMoE is enabled by its stochasticity based plasticity injection and sparse routing, which together counteract primacy bias and sustained neuron dormancy.

\begin{figure}[t]
    \centering
    \includegraphics[width=\linewidth]{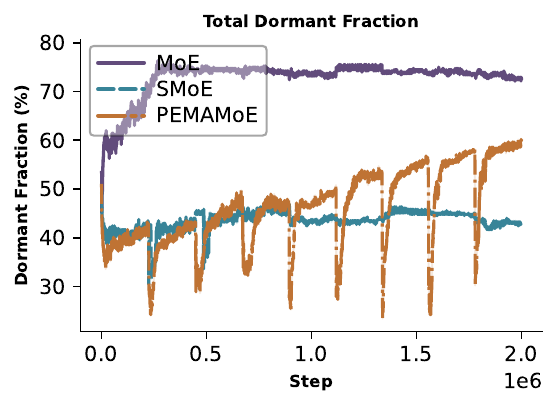}
    \caption{Evolution of the total dormant neuron fraction over training.}
    \label{fig:dormant}
\end{figure}

\subsubsection{Task performance and robustness under phase switches}
Fig.~\ref{fig:reward} compares the mean return across methods. PE\text{-}MAMoE delivers the highest asymptotic returns in most phases and the smoothest recovery after each abrupt switch, whereas MoE and SMoE suffer deeper dips and slower rebounds. This aligns with the theory that conditional computation plus controlled stochasticity bolsters adaptation while avoiding long term plasticity loss.

\begin{figure}[t]
    \centering
    \includegraphics[width=\linewidth]{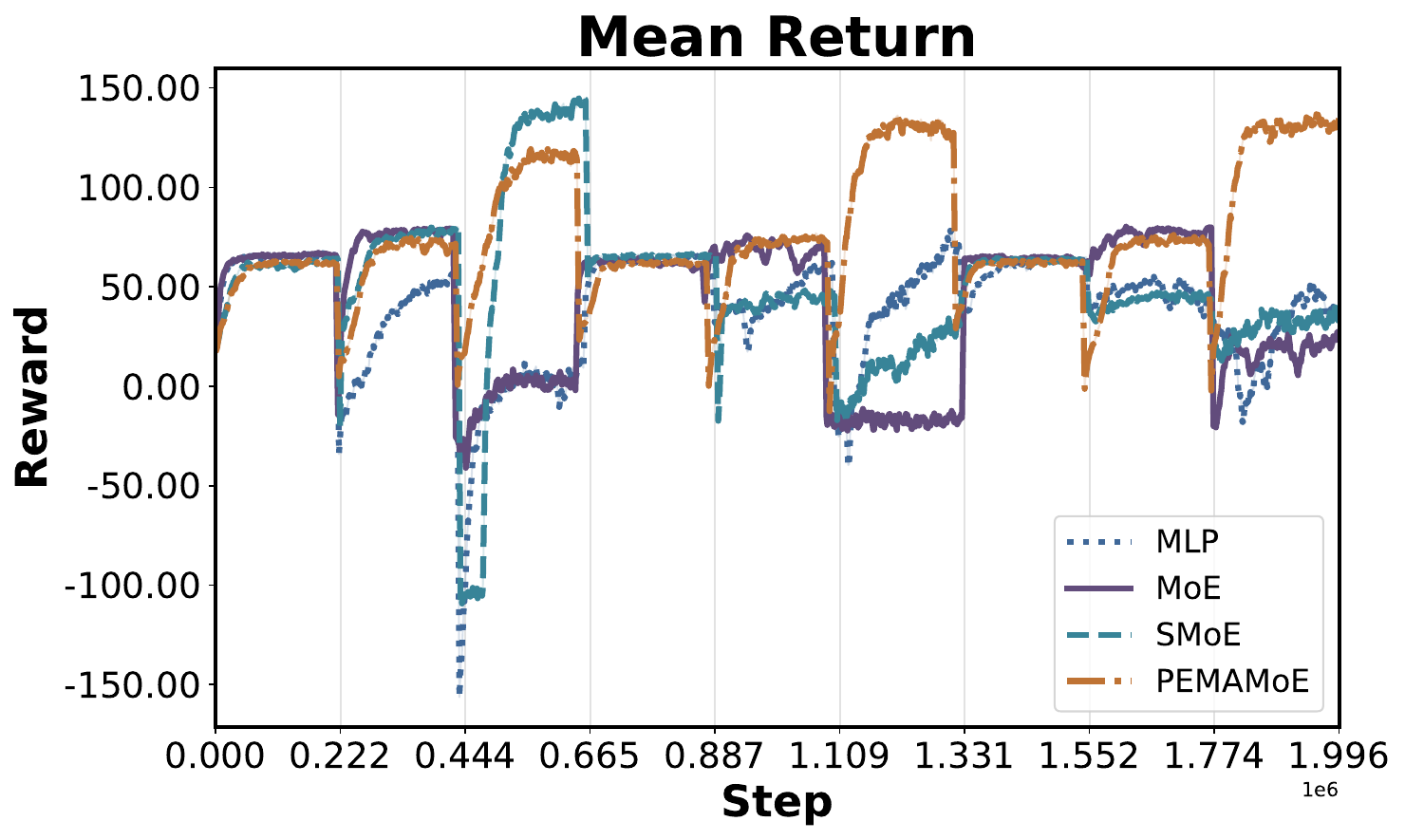}
    \caption{Mean Return over Training.}
    \label{fig:reward}
\end{figure}

To summarize performance across multiple objectives, Fig.~\ref{fig:uav_bar} reports the normalized interquartile mean computed over five metrics: coverage, return, and served users, as well as energy consumption and collision rate. Following the robust evaluation protocol of Agarwal \emph{et al.}~\cite{Agarwal2021Rliable}, IQM is computed by sorting the per-episode metric values recorded during training, discarding the bottom 25\% and the top 25\%, and averaging the remaining middle 50\%; this trims outlier episodes and gives a more reliable central tendency than the ordinary mean. Each method is trained with a single long run under the same random seed to ensure identical environment stochasticity across comparisons. For cross-metric visualization in Fig.~\ref{fig:uav_bar}, we apply per-metric min--max normalization: for each metric $m$, let $v_{\min}^{m}$ and $v_{\max}^{m}$ be the smallest and largest IQM values among all methods, then the normalized score of method $a$ is $({\mathrm{IQM}}_a^{m} - v_{\min}^{m}) / (v_{\max}^{m} - v_{\min}^{m})$. This rescales every metric to $[0,1]$, allowing direct comparison of quantities with different units. Error bars show standard error of the IQM computed across episodes. PE\text{-}MAMoE dominates or ties the best on most metrics: it achieves the highest normalized IQM on Coverage, Return, and ServedNumber, matches the best EnergyUse, and yields the lowest Collisions. Combined with the smallest return dips and fastest post switch recovery in Fig.~\ref{fig:reward}, these results indicate that PE\text{-}MAMoE is the most robust method under regime jumps, validating its design goal: maintain high performance while preserving plasticity across dynamic objectives.

Across plasticity proxies, including a high feature rank that recovers periodically and phase synchronous troughs in DNF, as well as end task metrics such as return and normalized IQM, PE\text{-}MAMoE consistently adapts the fastest and stabilizes at the highest level after objective switches. This provides converging evidence that sparse MoE routing combined with annealed stochasticity injection mitigates dormant neuron accumulation and plasticity loss, yielding state of the art robustness for non-stationary UAV--ECN control.

\begin{figure}[t]
    \centering
    \includegraphics[width=\linewidth]{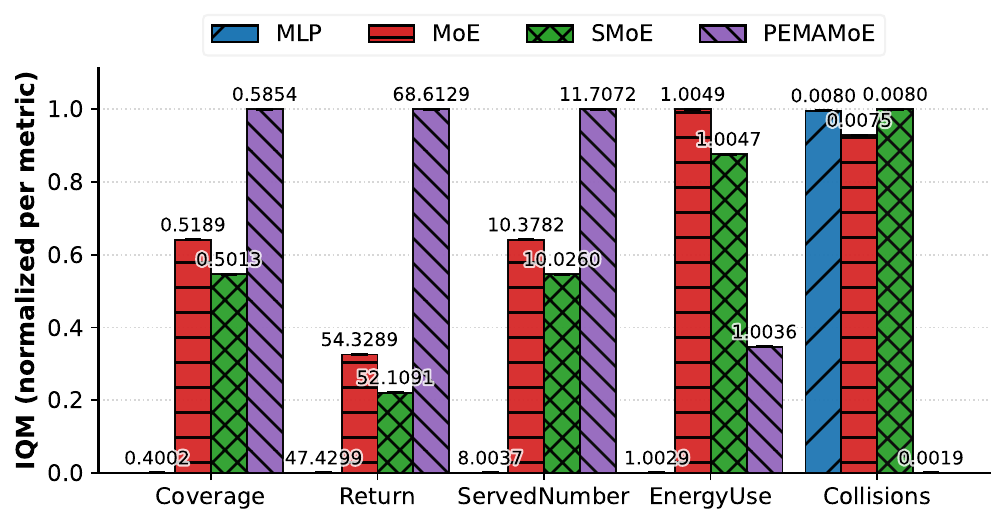}
    \caption{Normalized IQM across five system metrics. Each bar shows the per-metric min--max normalized IQM (middle 50\% of episode values); error bars indicate standard error. Higher is better for Coverage, Return, and ServedNumber; lower is better for EnergyUse and Collisions (plotted as $1-\text{normalized value}$ so that taller bars remain preferable).}
    \label{fig:uav_bar}
\end{figure}

\subsection{Selection Probabilities of Each Expert}
\label{subsec:expert_selection}

\paragraph{Policy router}
Fig.~\ref{fig:router_policy} plots the routing rate, the probability that the policy router selects a given expert, over training. PE\text{-}MAMoE reveals two clear behaviors. First, its selection probabilities remain well spread across experts rather than collapsing onto a single expert, indicating sustained capacity usage and diversity. Second, the routing shows clear phase synchronous cycles: after each switch, the router transiently re-allocates traffic and then re-stabilizes to a new mixture, repeating this pattern at every subsequent switch.
This cyclic redistribution is consistent with our plasticity mechanism and with the design goal of conditional computation: adapt expert usage when objectives change, and consolidate once the new regime is learned.

\begin{figure}[t]
  \centering
  \subfloat[Policy-side expert routing probability]{
    \includegraphics[width=0.98\linewidth]{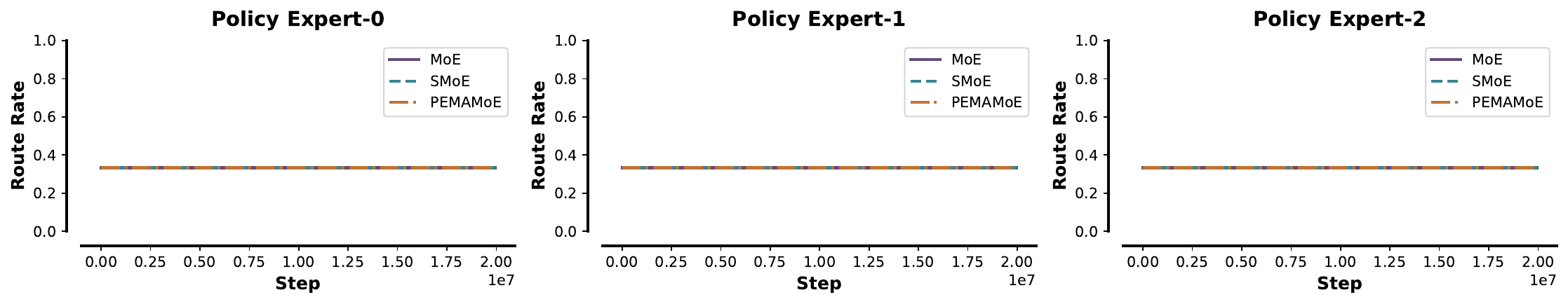}
    \label{fig:router_policy}
  }\\[0.6em]
  \subfloat[Value-side expert routing probability]{
    \includegraphics[width=0.98\linewidth]{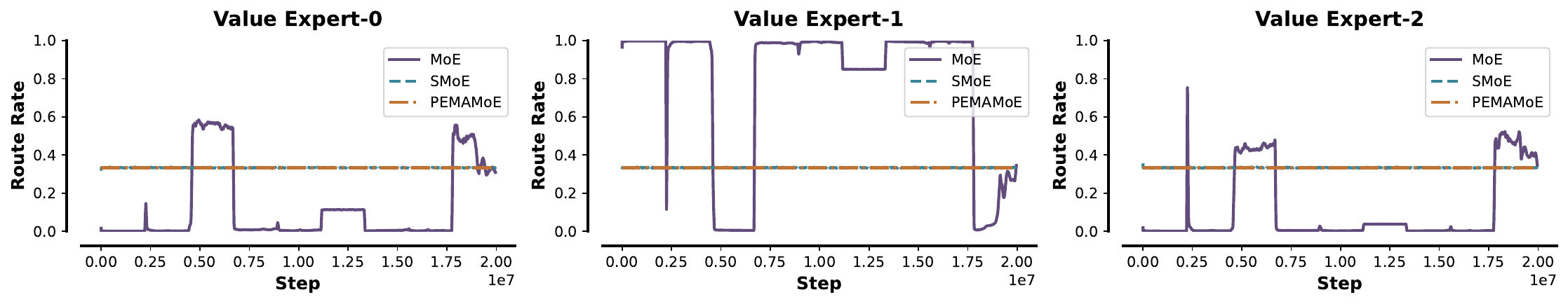}
    \label{fig:router_value}
  }
  \caption{Probability distribution of each expert selected for Policy and Value.}
  \label{fig:router}
\end{figure}

\paragraph{Value router}
Fig.~\ref{fig:router_value} shows the centralized critic router’s selection probabilities. The same trends emerge: PE\text{-}MAMoE maintains balanced, non-degenerate expert usage and rapidly reconfigures at phase boundaries, then settles to a stable allocation within each phase. This indicates that not only the actor but also the critic leverages specialized experts per regime, which helps stabilize advantages and speeds re-adaptation after switches. In contrast, MoE exhibits longer periods of over reliance on a subset of experts and slower post switch re-balancing.

Sustained diversity in routing probabilities is a necessary condition for avoiding the well known MoE failure modes of expert collapse and dropped tokens under load imbalance~\cite{Shazeer2017MoE, zhou2022expertchoice}. The observed periodic reallocation in PE\text{-}MAMoE is precisely the expected outcome of combining sparse MoE routing with controlled stochasticity: noise in the gate encourages exploration and prevents early commitment, while load balancing and entropy regularization preserves capacity across experts. The router’s quick, phase synchronous shifts are consistent with resetting or perturbation remedies for primacy bias, facilitating fast “re-plasticization” at regime changes, then annealing back to stability~\cite{nikishin2022primacy}. Overall, the selection probability dynamics corroborate that PE\text{-}MAMoE actively reallocates expertise when objectives change while retaining balanced utilization within each phase, a prerequisite for the superior robustness observed in return and IQM metrics.

\paragraph{Top-1 routing and expert diversity}
A legitimate concern with top-$1$ ($k{=}1$) routing is that a single expert may dominate and the remaining experts become idle. Fig.~\ref{fig:router_policy} directly addresses this: across the entire training run, PE-MAMoE's per-expert selection probabilities remain distributed among all three experts rather than collapsing onto one. Two mechanisms contribute. First, the noisy top-$k$ gating (Section~\ref{subsec:moe_arch}) adds softplus-scaled Gaussian noise to the router logits, so that even when one expert has the highest mean logit, stochastic tie-breaking keeps rival experts engaged. Second, the temperature schedule (Eq.~\eqref{eq:router-temp-only}) raises $\tau$ immediately after each switch, flattening the softmax and broadening the selection distribution when diversity matters most. In contrast, SMoE, which uses the same top-$1$ rule but lacks the plasticity mechanisms, shows visible periods of single-expert dominance (Fig.~\ref{fig:router_policy}), confirming that top-$1$ routing alone is insufficient and that the diversity is maintained by the combined plasticity design rather than by the routing rule itself. Regarding overhead, Table~\ref{tab:cost} shows that the additional parameter cost of maintaining three experts is modest ($3.1\times$ MLP), and top-$1$ activation keeps the per-step inference cost at only $1.16\times$ MLP. Scaling to larger expert pools ($E{>}3$) or higher $k$ values is a natural extension; we expect that load-balancing regularizers from the large-scale MoE literature~\cite{fedus2022switch,zhou2022expertchoice} will remain effective, and we flag this as an avenue for future work.

\subsection{System Performance}
\label{subsec:system_performance}

\paragraph{Safety under regime switches}
Fig.~\ref{fig:collisions} shows the collisions per episode across training. PE\text{-}MAMoE maintains near zero collision rates throughout, with only shallow transients at phase boundaries. In contrast, MoE and SMoE exhibit pronounced collision spikes after several switches, and the plain MLP baseline occasionally becomes unstable for extended intervals. In multi-agent control, collision rate is a primary safety key performance indicator (KPI) and a canonical RL penalty term; thus the consistently low collision profile indicates that PE\text{-}MAMoE preserves safe coordination while adapting to non-stationary objectives.

\begin{figure}[t]
    \centering
    \includegraphics[width=\linewidth]{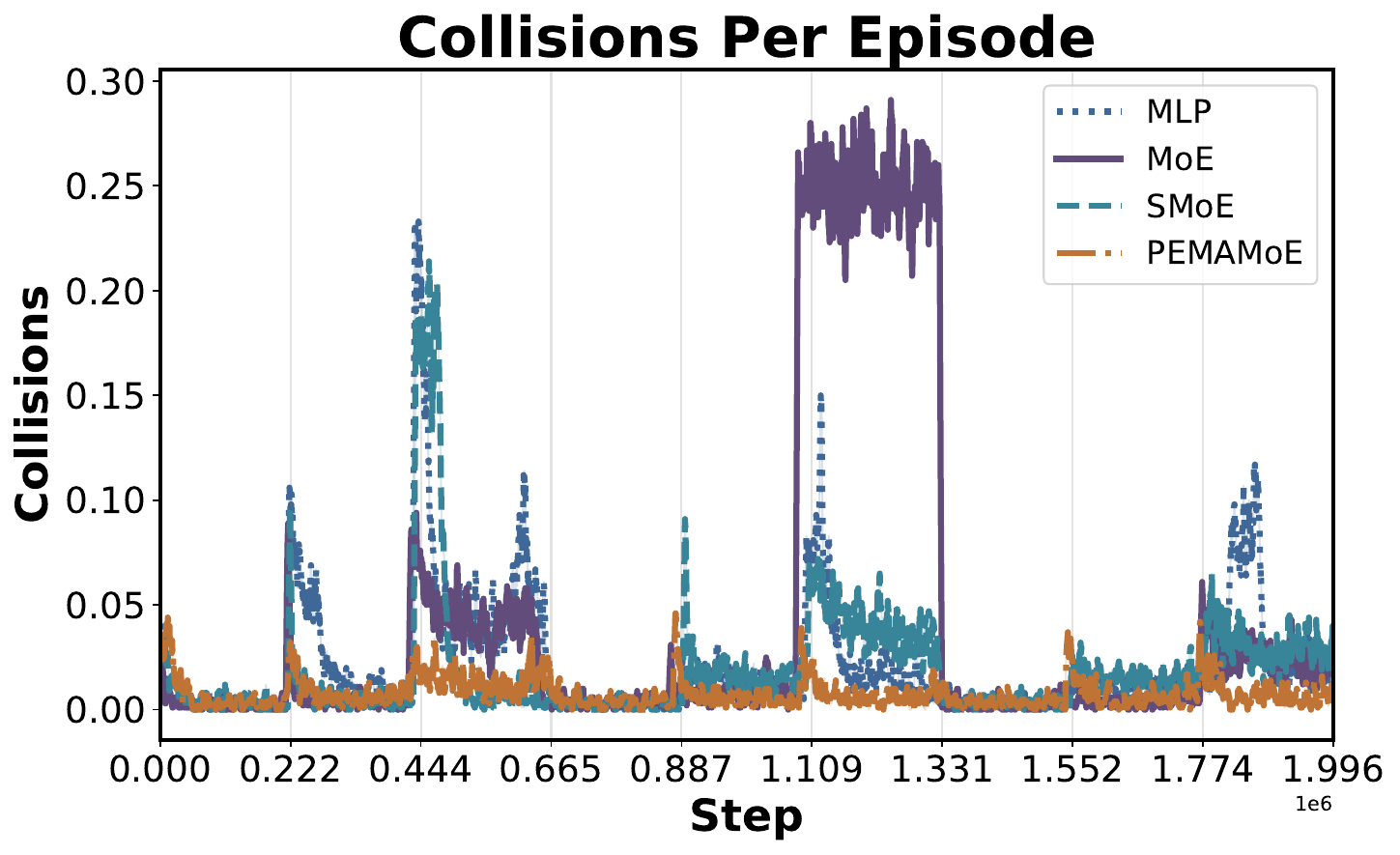}
    \caption{Collisions per Episode over Training.}
    \label{fig:collisions}
\end{figure}

\paragraph{Energy consumption and efficiency}
Fig.~\ref{fig:energy} reports the mean UAV energy rate. PE\text{-}MAMoE tracks the most energy efficient methods in steady regime and, crucially, avoids the severe energy dips experienced by MoE and SMoE when objectives switch. This matters because energy efficiency in wireless systems is fundamentally measured in bit/Joule (or its inverse Joule/bit), and transient collapses in rate or power control degrade the bits per Joule metric even if average power is similar. The smoother PE\text{-}MAMoE traces imply better energy per bit behavior during re-adaptation, consistent with our feasibility driven power allocation and the plasticity aware router that prevents large, wasteful corrections after phase changes.

\begin{figure}[t]
    \centering
    \includegraphics[width=\linewidth]{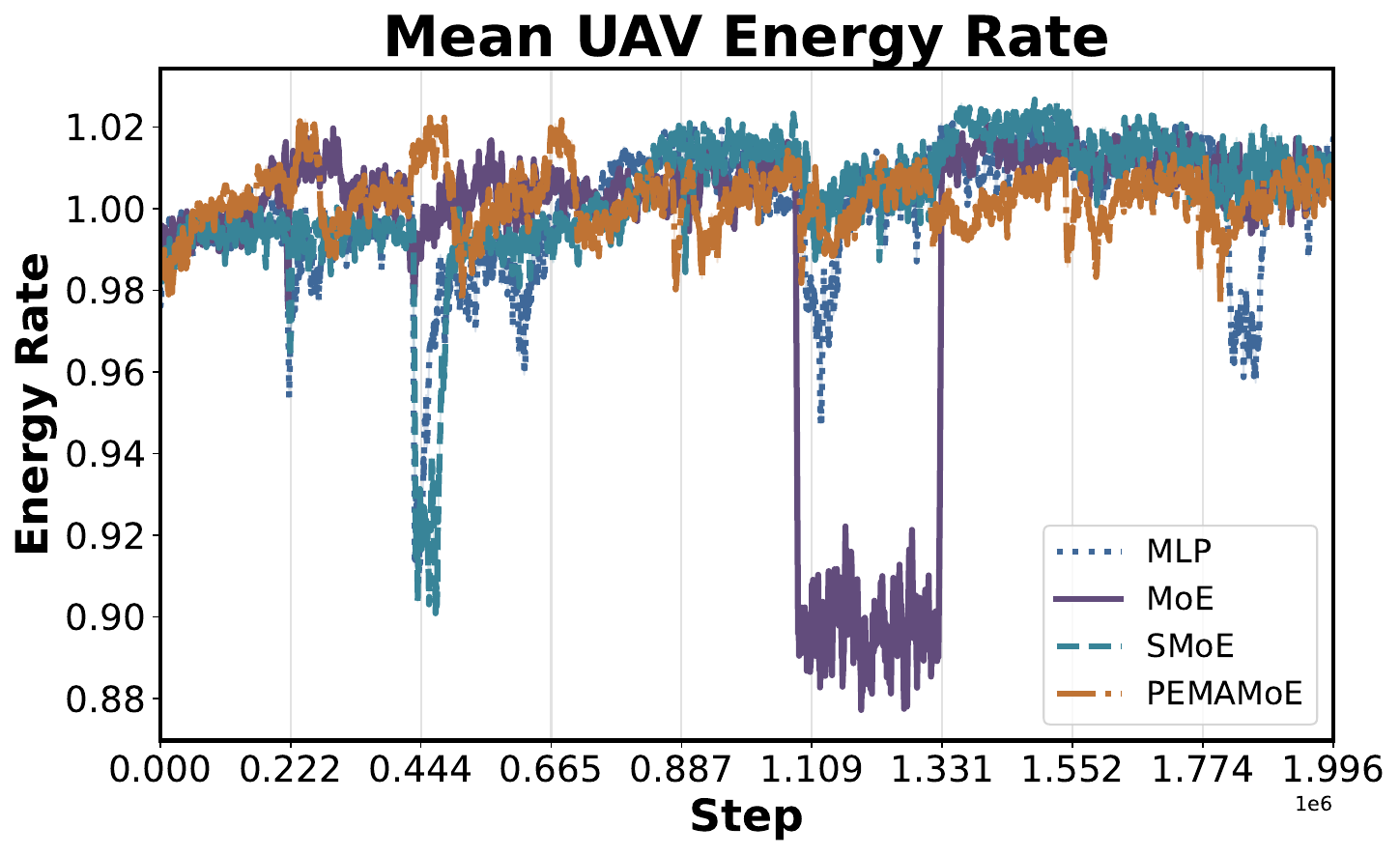}
    \caption{Mean UAV Energy Rate over Training.}
    \label{fig:energy}
\end{figure}

\paragraph{Service capacity and resource utilization}
Fig.~\ref{fig:sys_serve_users} reports the mean successfully served users; Fig.~\ref{fig:sys_uav_using_rate} shows the mean UAV utilization rate. Across phases, PE\text{-}MAMoE reaches the highest or tied highest stable number of served users and exhibits the fastest recovery after each abrupt switch. Simultaneously, its utilization remains high and stable, indicating effective load balancing across UAVs rather than overloading a subset. As served user count is a direct proxy for network capacity, these results, together with the energy profile, suggest that PE\text{-}MAMoE converts bandwidth and power into delivered service more reliably than competing methods.

\begin{figure}[t]
  \centering
  \captionsetup[subfloat]{labelformat=empty} 
  \subfloat[(a) Successfully Served Users]%
  {\includegraphics[width=0.98\linewidth]{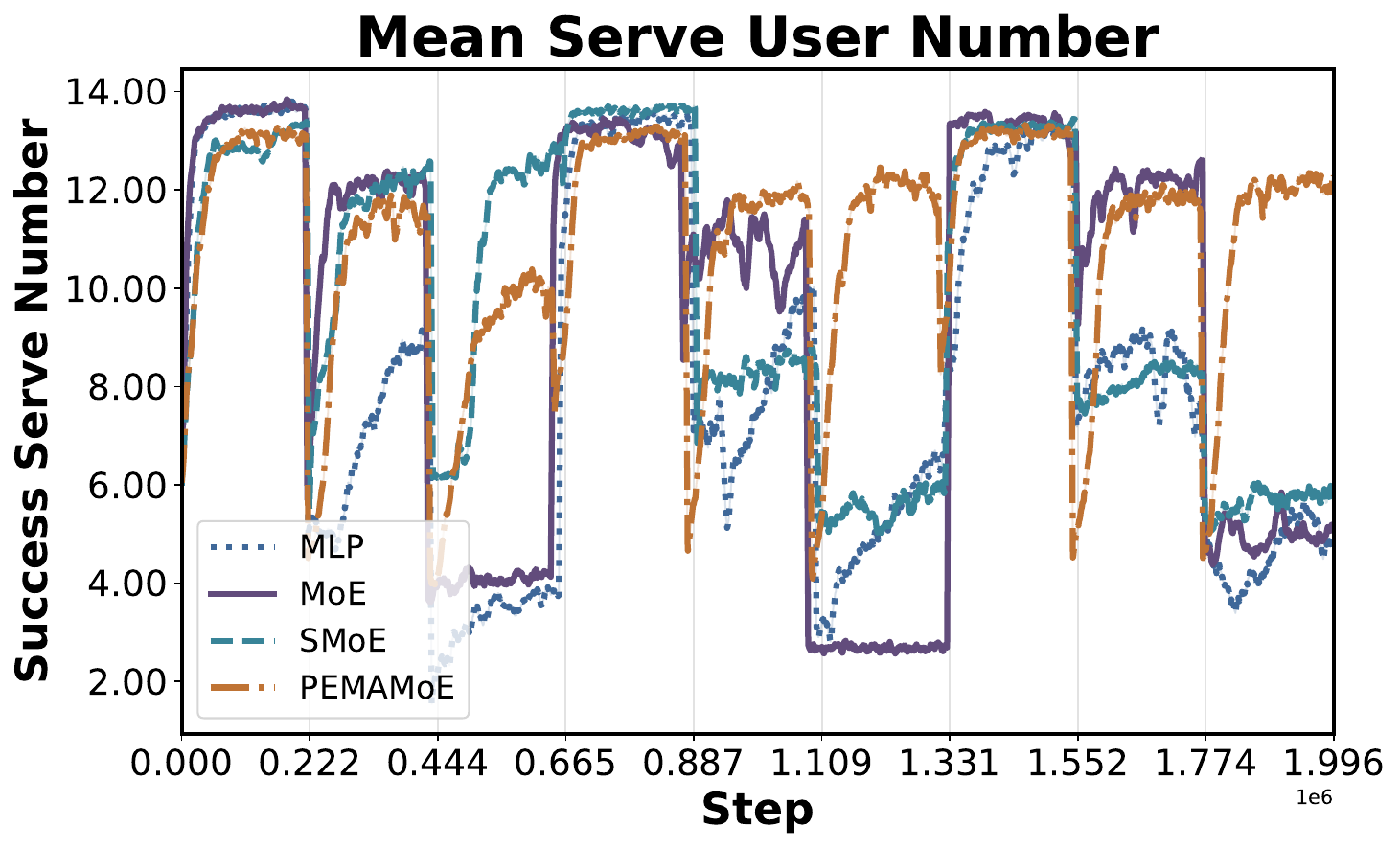}\label{fig:sys_serve_users}}\\[0.6em]
  \subfloat[(b) UAV Using Rate]%
  {\includegraphics[width=0.98\linewidth]{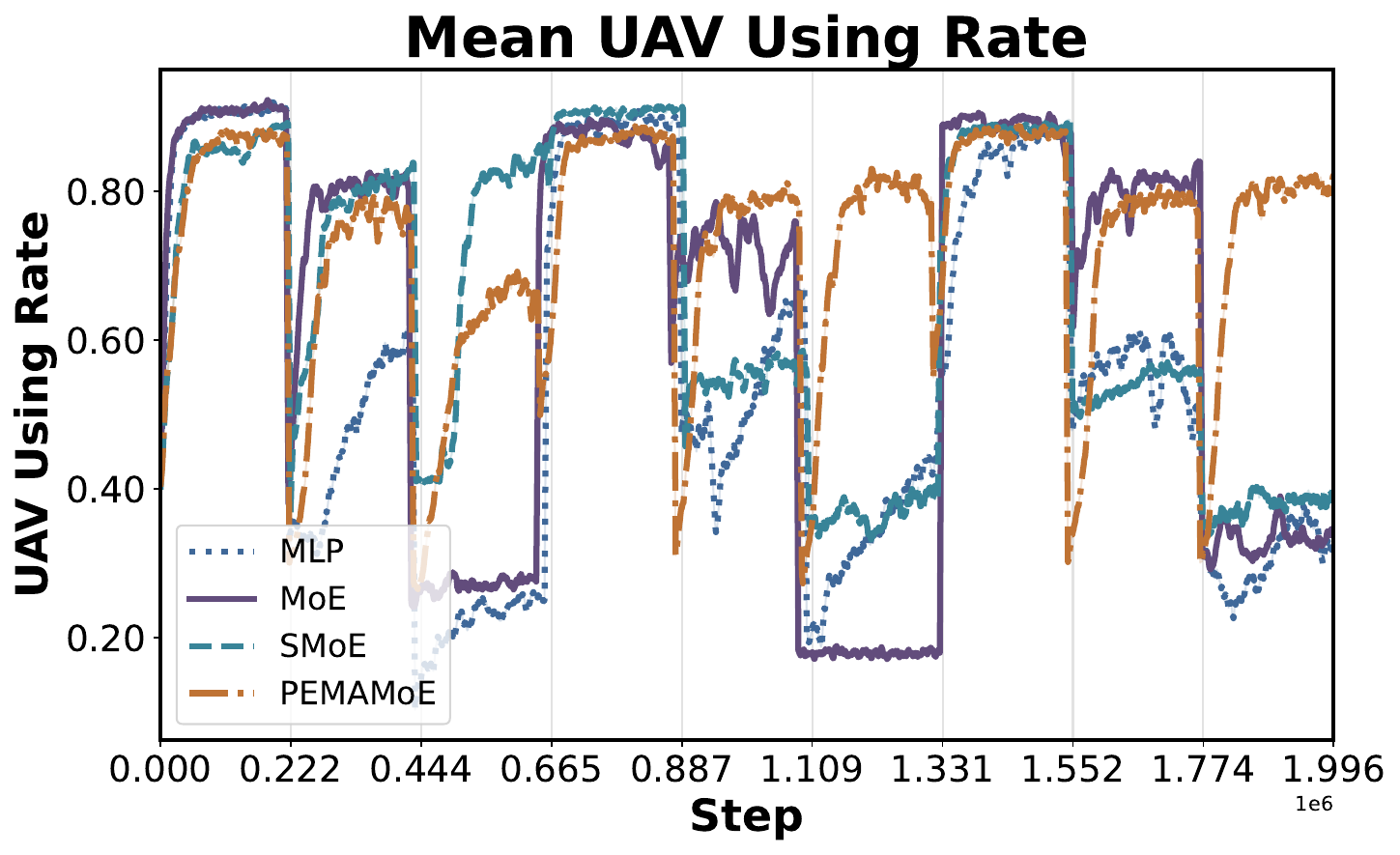}\label{fig:sys_uav_using_rate}}

  \caption{Performance comparison based on UAVs served users.}
  \label{fig:sys_metrics}
\end{figure}

Across safety, efficiency, and capacity, PE\text{-}MAMoE delivers the most consistent system level performance under phase driven non-stationarity. These outcomes corroborate our design: sparse MoE routing prevents expert collapse, while annealed plasticity injection enables quick re-plasticization at regime changes without sacrificing safety or energy efficiency.

\subsection{Influence on System Internal States}
\label{subsec:internal_states}

\paragraph{Action-level adaptation to regime switches}
Fig.~\ref{fig:action} reports the mean action over UAVs as phases change. PE\text{-}MAMoE reacts with clear but bounded action shifts exactly at the switch points, then quickly re-centers to a new operating level within each phase. In contrast, MLP exhibits large oscillations and prolonged overshoot, while vanilla MoE under reacts in some phases and SMoE often shows delayed, high variance corrections. These curves indicate that PE\text{-}MAMoE achieves a good balance between responsiveness and stability: making decisive adjustments when the target changes, followed by stable control with minimal jitter.

\begin{figure}[t]
    \centering
    \includegraphics[width=\linewidth]{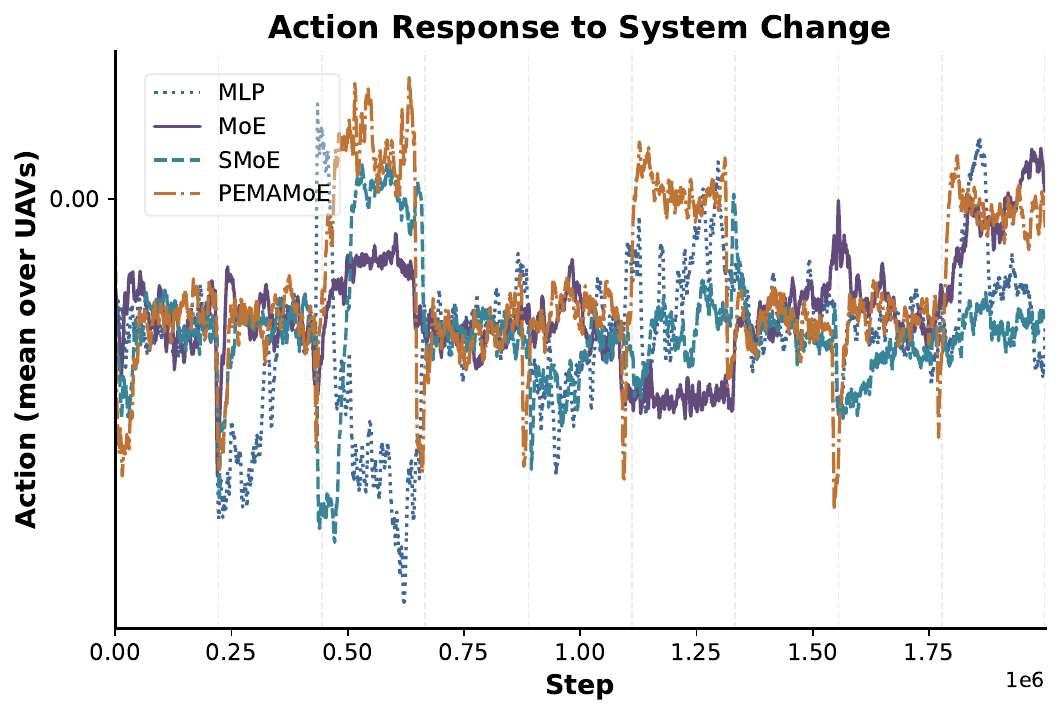}
    \caption{Mean Action Response to Phase Switches.}
    \label{fig:action}
\end{figure}

\paragraph{Why the actions look different}
Two design elements explain PE\text{-}MAMoE’s internal behavior. First, PPO style clipped updates with a centralized critic provide stable policy steps after each switch, avoiding the large policy jumps observed for baselines that either over commit or mis-estimate advantages. Second, sparsely gated MoE with load balancing allows the router to reallocate traffic to the most relevant experts for the new phase, rather than forcing a single monolithic network to change all at once; this conditional computation supports fast reconfiguration without global instability. The brief noise injection at switches counteracts primacy bias, enabling re-plasticization without sustained action noise.

\paragraph{Link to higher level performance}
The action dynamics match system outcomes reported earlier: low collision spikes (Fig.~\ref{fig:collisions}), smooth energy usage (Fig.~\ref{fig:energy}), and rapid recovery in served users and returns (Figs.~\ref{fig:sys_serve_users} and Fig.~\ref{fig:reward}). In reliable RL evaluation, such stability after distribution shifts is reflected by robust aggregates like IQM rather than single run peaks~\cite{Agarwal2021Rliable}. The observed PE\text{-}MAMoE traces moderate transients at switches and low variance within phases, explain its superior IQM and return, corroborating that sparse MoE routing \& plasticity injection steers internal states toward fast adaptation and stable control.

\subsection{Ablation Study}
\label{subsec:ablation_study}
\subsubsection{Noise Sensitivity Analysis}

We ablate the plasticity injection level by varying the expert noise threshold $\gamma$ and reporting average number of concurrently served users after convergence in each regime. Unless otherwise stated, all other experiments in this paper use $\gamma=0.005$.

\paragraph{Observations (Fig.~\ref{fig:noise})}
PE\text{-}MAMoE exhibits a clear sweet spot: coverage improves as $\gamma$ increases from $0.001$
to $0.005$ and then collapses when $\gamma$ is raised to $0.010$. In contrast, vanilla MoE remains essentially flat across thresholds. The non-monotonic curve for PE\text{-}MAMoE reflects the expected trade off: (i) too little noise fails to overcome primacy bias and dormant neuronal states, yielding under adaptation after phase changes; (ii) moderate noise re-plasticizes the router and experts, enabling rapid specialization and higher steady state coverage; (iii) excessive noise injects instability and harms policy improvement, reducing coverage. This pattern is consistent with prior evidence that periodic perturbations combat primacy bias in DRL, while overly large perturbations degrade performance, and with continual learning results showing that shrink and perturb style variability mitigates plasticity loss but should be carefully tuned. The flat MoE curve further supports that conditional computation alone is insufficient without plasticity maintenance mechanisms.

\begin{figure}[t]
    \centering
    \includegraphics[width=\linewidth]{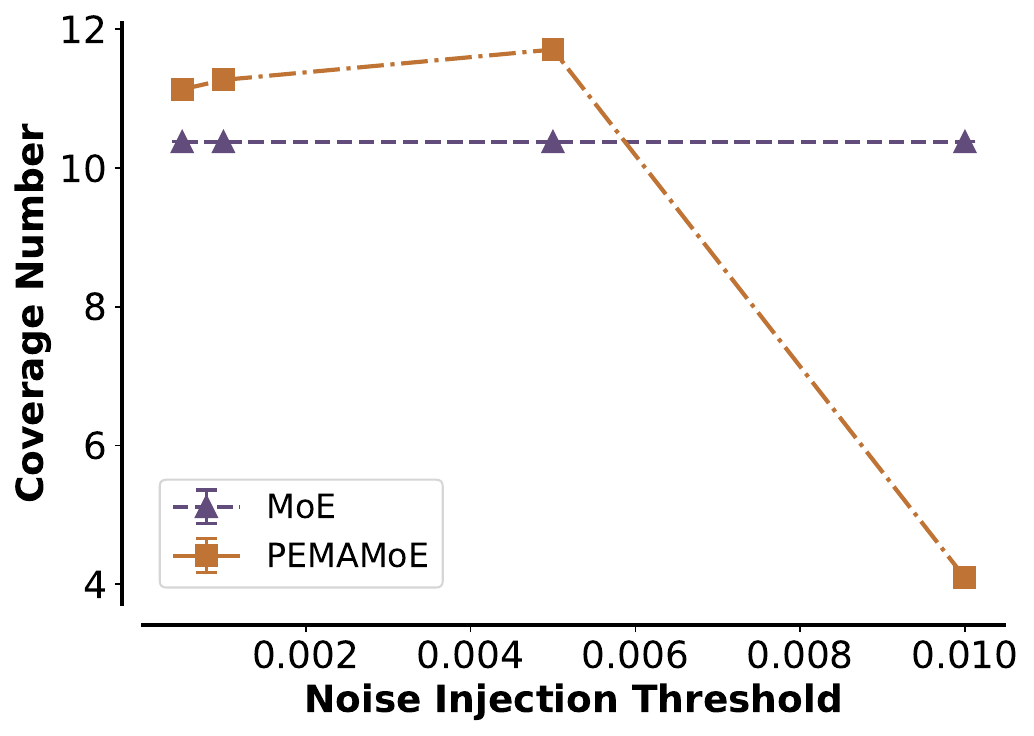}
    \caption{Effect of Noise Injection Threshold on Coverage.}
    \label{fig:noise}
\end{figure}

\subsubsection{Tricks Analysis}
\label{subsubsec:tricks}

Fig.~\ref{fig:tricks} reports the IQM drop in mean return when we remove one component from PEMAMoE at a time. We highlight four implementation choices below.

\begin{figure}[t]
    \centering
    \includegraphics[width=\linewidth]{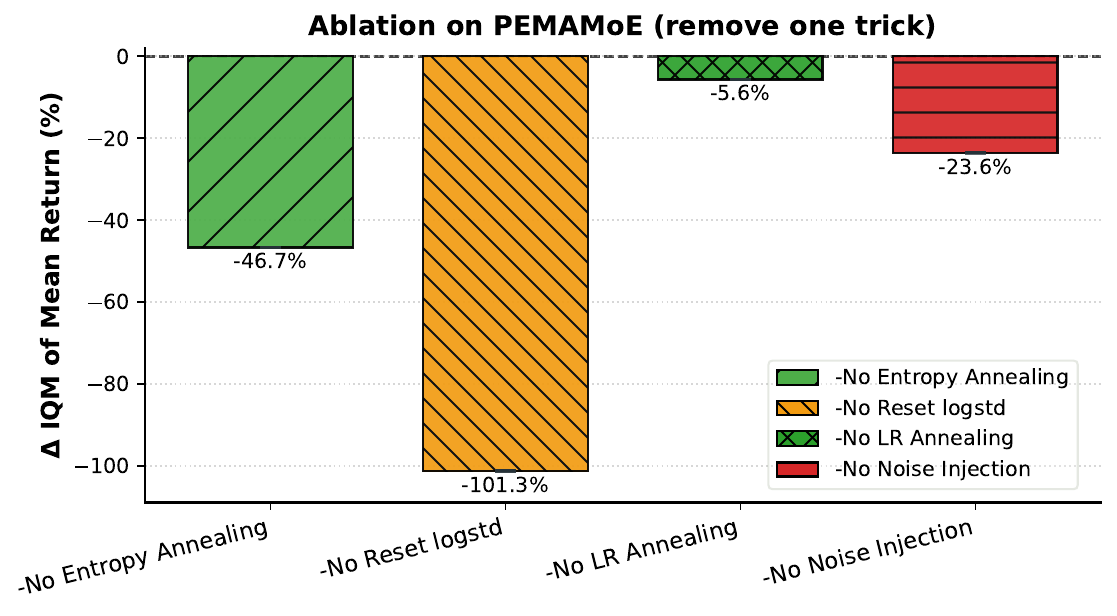}
    \caption{Ablation Study: IQM Drop When Removing Each Component.}
    \label{fig:tricks}
\end{figure}

\paragraph{Entropy Annealing}
We linearly anneal the policy’s entropy bonus within each phase, and reset the entropy coefficient to its initial value right after a phase switch before re-annealing. Removing this schedule causes a substantial performance drop ($-46.7\%$ IQM). Intuitively, early phase entropy helps exploration after regime changes; annealing then tightens the policy for exploitation. 

\paragraph{Reset logstd at switches}
Our continuous action policy resets the Gaussian logstd immediately after each phase switch, then lets it relearn. Without this reset, IQM collapses by $-101.3\%$, the worst among ablations. Resetting combats primacy bias, an early over commitment that hampers later adaptation, by re-opening exploration after abrupt distribution shifts. This matches prior evidence that periodically re-initializing parts of the network helps the agent ``forge" stale priors and adapt faster.

\paragraph{Learning rate (LR) Annealing}
We anneal the LR within each phase and reset them to initial values at switches before re-annealing. Removing LR annealing yields a smaller but non-negligible degradation ($-5.6\%$ IQM), aligning with the common use of LR schedules to stabilize late stage updates and improve convergence across tasks. In our non-stationary setting, resets at phase boundaries prevent under or over shoot caused by stale step sizes.

\paragraph{Noise Injection (plasticity)}
We inject controlled Gaussian noise into expert weights only for the first $N$ epochs after a switch (then off), with a conservative default $\gamma{=}0.005$ for this ablation. Disabling noise leads to a $-23.6\%$ IQM drop. Moderate perturbations re-plasticize experts and the gate, avoiding dormant specialists and enabling rapid re-specialization; excessive noise would harm stability, hence our short post switch window and small magnitude. This aligns with prior results showing that reset and perturb strategies mitigate primacy bias but must be dosed carefully.

All four components contribute, but Reset logstd and Entropy Annealing with reset are the most critical mechanisms for fast, stable re-adaptation after regime switches; Noise Injection adds a further, complementary gain by restoring plasticity in MoE experts. Together with LR annealing, these mechanisms keep PEMAMoE responsive right after a switch and conservative once a phase has stabilized, precisely the reactivity--stability balance needed in non-stationary UAV-ECNs.

\paragraph{Interaction effects among components}
The individual IQM drops sum to approximately $177\%$, far exceeding $100\%$, which indicates strong positive interaction among the four mechanisms rather than independent, additive contributions. This is expected because the components address overlapping but distinct facets of the same post-switch adaptation problem. Reset logstd and Entropy Annealing both re-open exploration after a switch, but through different channels: logstd reset directly widens the action distribution, while entropy annealing controls the policy-gradient regularizer. When both are present, each can operate at a moderate level; removing either forces the remaining one to compensate for the full exploration burden, which it cannot do alone, hence the super-additive drops. Noise Injection and Reset logstd are complementary rather than redundant: logstd reset re-opens action-space exploration (what the policy samples), whereas noise injection re-diversifies the expert feature representations (what the network can express). Removing noise injection while logstd reset is active still causes a $-23.6\%$ drop, confirming that action-level exploration alone does not prevent the underlying representation collapse that noise injection targets. LR Annealing interacts more weakly ($-5.6\%$), suggesting it plays a supporting role, stabilizing the gradient steps that the other three mechanisms make possible, rather than driving adaptation on its own. In summary, the four components form a layered system: logstd reset and entropy annealing handle \emph{behavioral} exploration, noise injection handles \emph{representational} plasticity, and LR annealing stabilizes the resulting updates. Their joint effect is a design feature, not an artifact, and reflects the multi-faceted nature of the post-switch adaptation challenge.

\section{Conclusions}
\label{sec:conclusion}

We presented PE\text{-}MAMoE, a plasticity enhanced multi agents mixture of experts framework for UAV-assisted emergency communications operating under phase driven non-stationarity. The method combines sparsely gated conditional computation with switch aware, fixed window expert noise injection to preserve plasticity, and adopts MAPPO for stable cooperative learning. Across extensive simulations grounded in 3GPP style A2G channel modeling, PE\text{-}MAMoE consistently achieved the best or near best performance on system level metrics. Internally, we observe phase synchronous signals, periodic peaks in expert feature rank and switch synchronous oscillations in dormant neuron fraction, indicating effective re-plasticization at regime switches without explicit dormant unit reinitialization. 

Although PE\text{-}MAMoE delivers consistent gains under phase driven non-stationarity and is grounded in standardized A2G propagation, several limitations remain and open concrete avenues for improvement. First, our simulator abstracts flight, sensing and actuator dynamics; validating with hardware in the loop or field trials is needed to assess latency, control robustness, and safety envelopes under wind and global navigation satellite system (GNSS) imperfections. Second, while the radio frequency (RF) stack follows 3GPP guidance, interference control is simplified; richer physical layer (PHY), medium access control (MAC) and tighter cross layer learning are left for future work. Third, the current CTDE setting assumes reliable global observations during training and no explicit inter-UAV messaging at execution; extending to communication constrained MARL with learned, bandwidth budgeted messaging is a natural step. Fourth, our plasticity mechanism uses switch aware, fixed window noise injection on experts only; learning the noise schedule online and jointly balancing adaptation and stability under unknown variation budgets remains open, especially in light of primacy bias phenomena in DRL.
Finally, scaling to larger fleets and expert pools raises system concerns; importing routing regularizers and systems techniques from large scale SMoE deployments, together with stronger reliability reporting beyond single point estimates, will be important for safety critical ECNs.
Fifth, our phase model cycles deterministically through three regimes with fixed period $L$. While this is a common benchmarking protocol for switching environments, it represents a relatively benign form of non-stationarity because, in principle, an agent could learn to anticipate switches. We note that (i) the phase index is \emph{not} provided to the agent; agents must infer the current regime solely from observed demand distributions and user spatial patterns, and (ii) the primary challenge studied here is not switch prediction but \emph{representation degradation} (rank collapse and dormant neurons) that accumulates over long training regardless of predictability. Nonetheless, evaluating PE-MAMoE under randomized phase durations, orderings, or entirely unseen regimes is a natural next step that would further test the robustness of the plasticity injection mechanism and is left for future work.

\bibliographystyle{IEEEtran}
\bibliography{reference}

\newpage

\begin{IEEEbiography}[{\includegraphics
[width=1in,height=1.25in,clip, keepaspectratio]{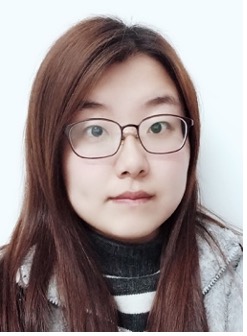}}]
{Wen Qiu} is currently a Postdoctoral Fellow in Kitami Institute of Technology. She received her Ph.D. degree in Co-creative Engineering from Kitami Institute of Technology, Japan. Her research interests include deep reinforcement learning and emergency wireless communication networks, focusing on developing intelligent networking solutions for disaster response scenarios. Her work specifically addresses challenges in network resilience, dynamic resource allocation, and system optimization using advanced machine learning techniques.
\end{IEEEbiography}

\begin{IEEEbiography}[{\includegraphics
[width=1in,height=1.25in,clip,
keepaspectratio]{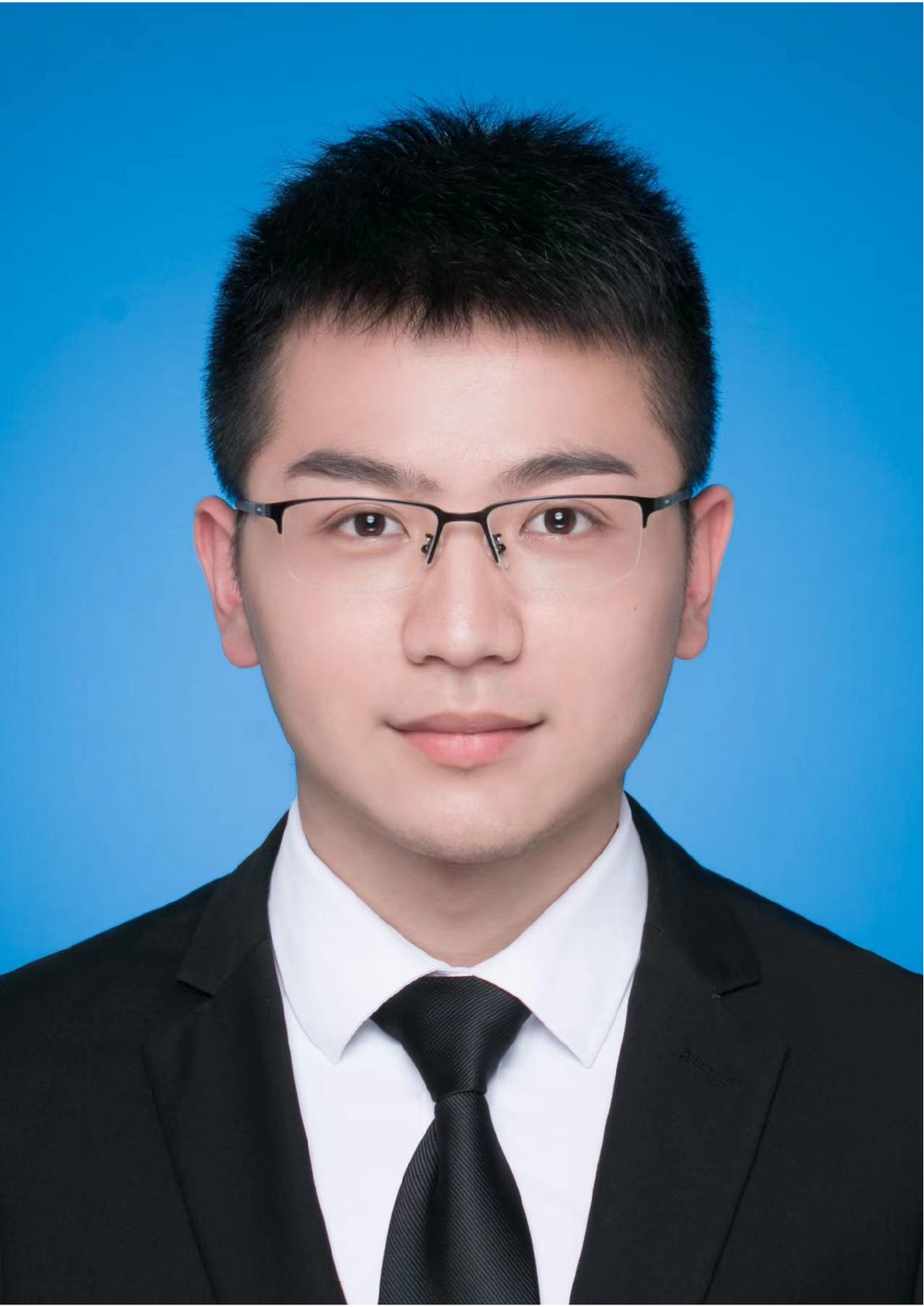}}]
{Zhiqiang He} is currently pursuing a Ph.D. in the University of Electric Communications in Japan. He received his MS degree in Control Science and Engineering from Northeastern University, Shenyang, China. His research interests focus on deep reinforcement learning and its applications. He previously worked at Baidu and InspirAI, where he developed a master-level AI for the game Landlord that outperformed professional players. 
\end{IEEEbiography}

\begin{IEEEbiography}
[{\includegraphics[width=1in,height=1.25in,clip,keepaspectratio]{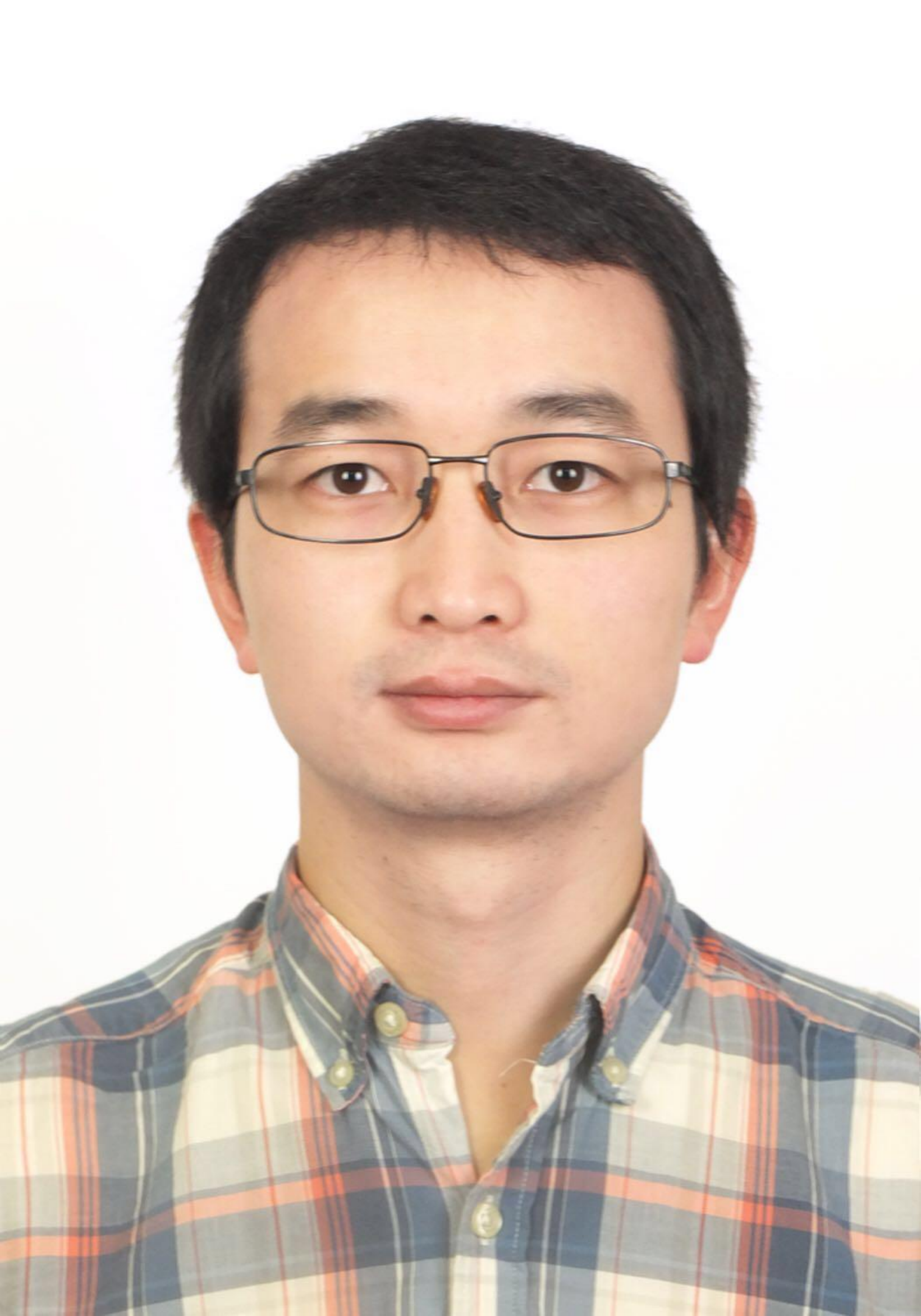}}]
{Wei Zhao} (S’12-M’16) received his Ph.D. degree in the Graduate School of Information Sciences, Tohoku University. He is currently a Professor at the School of Computer Science and Technology, Anhui University of Technology. His research interests include deep reinforcement learning, edge computing, and resource allocation in wireless networks. He was the recipient of the IEEE WCSP-2014 Best Paper Award, and IEEE GLOBECOM-2014 Best Paper Award. He is a member of IEEE.
\end{IEEEbiography}

\begin{IEEEbiography}[{\includegraphics[width=1in,height=1.25in,clip,keepaspectratio]{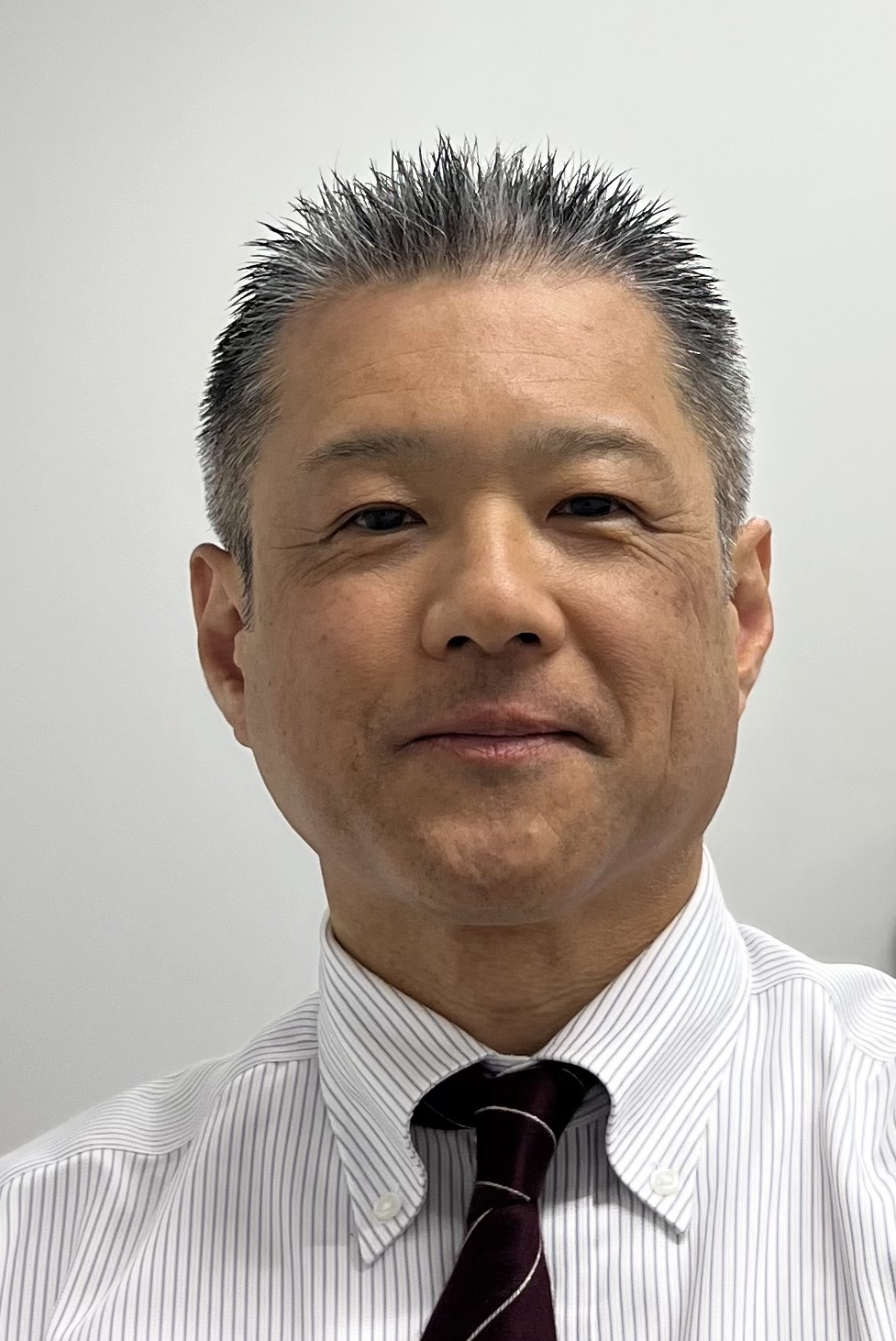}}]{Hiroshi Masui}
received his Ph.D. in Science from Osaka University in 1998. His research interests are theoretical nuclear physics, analysis on public transportation, cloud optimization, data-driven science. He currently works on Information Processing Center of Kitami Institute of Technology.
\end{IEEEbiography}

\end{document}